\documentclass[12pt, draftclsnofoot, onecolumn]{IEEEtran}

\pdfoutput=1
\usepackage[noadjust]{cite}
\usepackage{graphicx}
\usepackage{amsmath}
\usepackage{enumerate}
\usepackage{amssymb}
\usepackage{fixltx2e}
\usepackage{color}
\usepackage{algorithm}
\usepackage{algorithmic}
\usepackage{multirow}
\usepackage{ragged2e}
\allowdisplaybreaks[4]
\usepackage{mathtools}
\usepackage{bm}
\newtheorem{prop}{Proposition}
\begin{document}

\title{\LARGE{Channel Estimation and Multipath Diversity Reception for RIS-Empowered Broadband Wireless Systems Based on Cyclic-Prefixed Single-Carrier Transmission}}
\author{Qiang~Li, \IEEEmembership{Member,~IEEE}, Miaowen~Wen, \IEEEmembership{Senior Member, IEEE}, Ertugrul Basar, \IEEEmembership{Senior Member, IEEE}, George C. Alexandropoulos, \IEEEmembership{Senior Member, IEEE}, Kyeong Jin Kim, \IEEEmembership{Senior Member, IEEE}, and H. Vincent Poor, \IEEEmembership {Life Fellow, IEEE}
	\thanks{Qiang~Li is with the College of Information Science and Technology, Jinan University, Guangzhou 510632, China (e-mail: qiangli@jnu.edu.cn).}
	\thanks{Miaowen~Wen is with the School of Electronic and Information Engineering, South China University of Technology, Guangzhou 510640, China (e-mail: eemwwen@scut.edu.cn).}
	\thanks{Ertugrul Basar is with the Department of Electrical and Electronics Engineering, Ko\c{c} University, Istanbul 34450, Turkey (e-mail: ebasar@ku.edu.tr).}
	\thanks{George C. Alexandropoulos is with the National and Kapodistrian University of Athens, Panepistimiopolis Ilissia, Greece  (e-mail:  alexandg@di.uoa.gr).}
	\thanks{Kyeong Jin Kim is with the Mitsubishi Electric Research Laboratories, Cambridge, MA 02139, USA (e-mail: kkim@merl.com).}
	\thanks{H. Vincent Poor is with the Department of Electrical and Computer Engineering, Princeton University, Princeton, USA (e-mail: poor@princeton.edu).}
}
\maketitle

\begin{abstract}
In this paper, a cyclic-prefixed single-carrier (CPSC) transmission scheme with phase shift keying (PSK) signaling is presented for broadband wireless communications systems empowered by a reconfigurable intelligent surface (RIS). In the proposed CPSC-RIS, the RIS is configured according to the transmitted PSK symbols such that different cyclically delayed versions of the incident signal are created by the RIS to achieve multipath diversity. A practical and efficient channel estimator is developed for CPSC-RIS and the mean square error of the channel estimation is expressed in closed-form. We analyze the bit error rate (BER) performance of CPSC-RIS over frequency-selective Nakagami-$m$ fading channels. An upper bound on the BER is derived by assuming the maximum-likelihood detection. Furthermore, by resorting to the concept of index modulation (IM), we propose an extension of CPSC-RIS, termed CPSC-RIS-IM, which enhances the spectral efficiency. In addition to conventional constellation information of PSK symbols, CPSC-RIS-IM uses the full permutations of cyclic delays caused by the RIS to carry information. A sub-optimal receiver is designed for CPSC-RIS-IM to aim at low computational complexity. Our simulation results in terms of BER corroborate the performance analysis and the superiority of CPSC-RIS(-IM) over the conventional CPSC without an RIS and orthogonal frequency division multiplexing with an RIS.
\end{abstract}


\IEEEpeerreviewmaketitle

\section{Introduction}
As an alternative to orthogonal frequency division multiplexing (OFDM), cyclic-prefixed single-carrier (CPSC) shares some OFDM advantages, such as low-complexity implementation, while avoiding several OFDM drawbacks, such as a high peak-to-average power ratio and high sensitivity to carrier frequency offsets \cite{PancaldiSingle,FalconerFrequency}. Moreover, uncoded CPSC with a minimum mean square error (MMSE) receiver is able to extract multipath diversity in frequency-selective fading channels for reasonably large values of block sizes and practical values of bit error rate (BER) \cite{DevillersAbout}, and multipath diversity in identical \cite{KimPerformance1} and non-identical frequency selective fading channels \cite{KimDiversity} can be achieved by a maximum-likelihood (ML) receiver irrespective of the uncoded block size. By contrast, without channel coding or precoding techniques, OFDM cannot harvest any multipath diversity even with the optimal ML receiver \cite{DevillersAbout}. CPSC is a promising solution to broadband wireless communications, and has been included in the 3GPP Long Term Evolution standard for uplink transmission \cite{PancaldiSingle}. 

Rich multipath components, however, are required in the propagation environment for CPSC to achieve a significant gain in terms of BER. Unfortunately, this requirement may not be satisfied in general. Cyclic delay diversity (CDD) \cite{DammannStandard} is a simple yet effective technique to solve this problem. By transmitting the same signal from multiple antennas equipped in a single transmitter with different cyclic delays, a multiple-input channel is equivalently transformed to a single-input one with increasing multipath diversity and without altering the receiver structure. For CDD-CPSC, the zero-forcing (ZF) receiver fails to obtain any diversity gain, while the ML and MMSE receivers are able to pick up diversity gains under some conditions, as pure CPSC without CDD. To maximize the performance gain of applying CDD to CPSC systems, non-linear equalizers, such as the frequency-domain Turbo equalizer \cite{KwonCyclic} and the block iterative generalized decision feedback equalizer \cite{LiangDesign}, were developed. In cooperative CPSC systems, distributed CDD can be implemented via multiple single-antenna transmitters \cite{KimPerformance}. The idea of distributed CDD was further applied to physical-layer security enhancement \cite{KimSecrecy} and spectrum sharing systems \cite{IradukundaOn}. It is worth noting that all of the above-mentioned CDD transmission schemes are based on multiple local or distributed transmit antennas, each equipped with a radio frequency (RF) chain. This configuration obviously complicates the system implementation, expands the hardware cost, and increases the power consumption, which may be unfavorable for energy-constrained and size-limited devices.

Recently, reconfigurable intelligent surface (RIS)-assisted transmission has emerged as an easy-to-implement, low-cost, and green communication technique \cite{BasarWireless,RenzoSmart,JianReconfigurable}. On an RIS, there are a large number of passive reflecting elements, each of which is able to reflect and exert adjustable amplitude-phase changes on incident signals. In this sense, the amplitude-phase responses of RISs and associated channels can be reconfigured intentionally to achieve different purposes \cite{StrinatiReconfigurable}. By exploiting this property, RISs have been applied to enhance the wireless communications in terms of energy efficiency \cite{HuangReconfigurable,WuIntelligent}, weighted sum-rate \cite{GuoWeighted}, secrecy rate \cite{HongRobust,DuReconfigurable,AlexandropoulosSafeguarding}, network coverage \cite{YildirimHybrid}, error performance \cite{YeJoint}, localization accuracy \cite{AbuNear}, etc. In particular, an RIS can be considered as an RF chain-free multi-antenna device. Inspired by this observation, the authors in \cite{TangMIMO} and \cite{TangRealization} designed RIS-based multiple-input multiple-output transmission schemes for spatial multiplexing and Alamouti space–time block coding, respectively. On the other hand, the idea of index modulation (IM) \cite{BasarIndex}, which uses the indices of some resources/building blocks (e.g., antennas, subcarriers, time slots, and signal constellations) to convey information, has been introduced into RIS-aided communications. Obviously, RISs can be explicitly deployed to enhance the existing IM schemes \cite{CanbilenReconfigurable,LiSpace,LuoSpatial}. Moreover, new RIS-based IM schemes that use the indices of reflection patterns to carry information were developed in \cite{LiSingle,GuoReflecting,YanPassive,BasarReconfigurable,DashPerformance,YuanReceive,GopiIntelligent,LinReconfigurable,LinReconfigurable2}. In contrast to the above-mentioned RIS-related works that focus on narrowband wireless communications over flat fading channels, the authors in \cite{YangIntelligent,ZhengIntelligent,LinAdaptive} studied RIS-aided OFDM broadband wireless communications, in which the channel state information (CSI) is estimated in the frequency domain without multipath diversity gains.

To the best of our knowledge, CPSC transmission techniques have not been investigated for RIS-aided broadband wireless communications. Moreover, there have not been any reports of time-domain channel estimation methods for RIS-aided broadband frequency-selective channels in the literature. The fact that an RIS acts as a multi-element reflector has also not been exploited to implement CDD as well. Since OFDM and CPSC have different pros and cons, CPSC-RIS can be considered as an alternative to OFDM-RIS. They can coexist and complement each other for broadband communications. Against this background, we design a CPSC transmission scheme for RIS-empowered broadband wireless systems in this paper. However, the CP occupies a non-negligible bandwidth and leads to a reduction of the spectral efficiency (SE), which is a notable problem especially for short-packet communications. To solve this problem, we further develop an extension of CPSC-RIS by resorting to the concept of IM. Specifically, the main contributions of this paper are summarized as follows:
\begin{itemize}
	\item A CPSC transmission scheme with CDD and phase shift keying (PSK) signaling is proposed for RIS-empowered broadband wireless systems to harvest multipath diversity gains. An efficient pilot-aided channel estimator, which is able to estimate the equivalent channel in the time domain via one transmission block, is developed for CPSC-RIS. In this sense, CPSC-RIS provides a new framework of channel estimation for RIS-aided broadband communications systems.

	\item The proposed CPSC-RIS scheme provide a flexible design. By adjusting the block size, CPSC-RIS can adapt to different degrees of variability in the CSI. In particular, CPSC-RIS with a small block size can be dedicated to short-packet communications at the cost of some diversity gains. On the other hand, compared with conventional antenna-based CDD-CPSC, CPSC-RIS uses the RIS to achieve CDD, avoiding multiple complicated and power-hungry RF chains.
	
	\item Both ML and ZF/MMSE detectors are designed for CPSC-RIS. We analyze the BER performance of CPSC-RIS over frequency-selective Nakagami-$m$ fading channels. An upper bound on the BER of CPSC-RIS is derived in closed-form by assuming the ML detection without channel estimation errors. From the BER analysis, the diversity order and the performance improvement of CPSC-RIS over CPSC is characterized.
	
	\item An extension of CPSC-RIS, termed CPSC-RIS-IM, is also proposed to improve the SE by resorting to the concept of IM. To be more specific, CPSC-RIS-IM encodes the cyclic delays caused by the RIS for conveying additional information. A low-complexity detector is then designed for CPSC-RIS-IM.
\end{itemize}

The remainder of this paper is organized as follows. Section~II describes the system model of CPSC-RIS, including the designs of channel estimator and signal detectors. The performance of CPSC-RIS with the ML detection is analyzed in Section~III, followed by one IM-empowered extension of CPSC-RIS in Section~IV. Section V presents the computer simulation results, and finally, Section~VI concludes the paper.

\textit{Notation:} Column vectors and matrices are denoted by lowercase and uppercase boldface letters, respectively. Superscripts $(\cdot)^*$, $(\cdot)^T$, and $(\cdot)^H$ stand for conjugate, transpose, and Hermitian transpose, respectively. The $ N \times N $ identity matrix is symbolized by $ \mathbf{I}_{N \times N} $. $\mathbf{1}_{1 \times N}$ and $\mathbf{0}_{1 \times N}$ denote an all-one matrix and an all-zero matrix of size $1 \times N$, respectively. $\mathrm{diag}\{\cdot\}$ transforms a vector into a diagonal matrix. $\Re\{\cdot\}$ and $\Im\{\cdot\}$ return the real and imaginary parts of a complex number, respectively. $j=\sqrt{-1}$ is the imaginary unit. $(\mathcal{C})\mathcal{N}(\mu,\sigma^2)$ represents the (complex) Gaussian distribution with mean $\mu$ and variance $\sigma^2$. The probability of an event and the probability density function (PDF) of a random variable are denoted by $\Pr(\cdot)$ and $p(\cdot)$, respectively. $E\{\cdot\}$ and $Var\{\cdot\}$ denote expectation and variance, respectively. $\|{\cdot}\|$ stands for the Frobenius norm. $Q(\cdot)$, $\Gamma(\cdot)$, and $\lfloor \cdot \rfloor$ represent the Gaussian $Q$-function, Gamma function, and floor function, respectively. $\mathrm{rank}(\cdot)$ and $\mathrm{Tr}(\cdot)$ denote the rank and trace of a matrix, respectively. $\odot$ and $\star$ represent the Hadamard product and linear convolution, respectively. $\angle(\cdot)$ denotes the phase of a complex number and $\mathrm{cir}(\cdot)$ denotes right circulant operation.

\section{System Model}
\begin{figure}[t]
	\centering
	\includegraphics[width=5.5in]{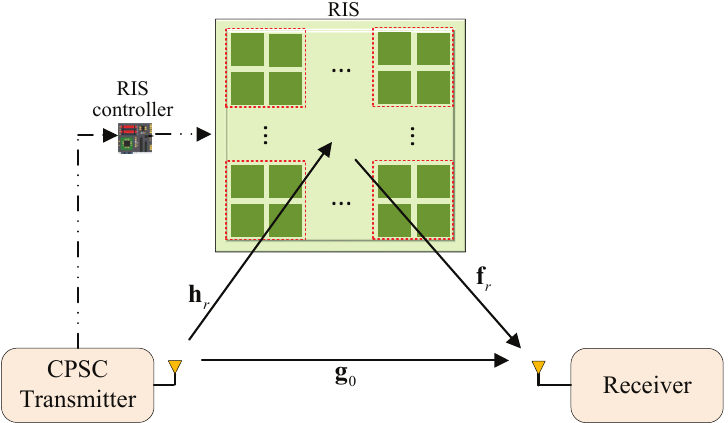}
	\caption{The block diagram of the proposed CPSC-RIS system comprising a single-antenna transmitter, a multi-element RIS, and a single-antenna receiver.}
	\label{Fig_viso}
\end{figure}
The block diagram of the proposed CPSC-RIS system is illustrated in Fig.~\ref{Fig_viso}, where a CPSC transmitter communicates with a receiver over frequency-selective block fading channels with the aid of an RIS. Both the transmitter and the receiver are equipped with a single antenna, while the RIS that is connected to a controller consists of $N_R$ reflecting elements. The RIS controller follows the instructions from the transmitter to adjust the reflection coefficients of the RIS \cite{StrinatiReconfigurable}. However, $N_R$ is typically a large number, which results in high overhead/complexity for channel estimation and coefficient adjustment. Hence, we adopt the grouping method \cite{YangIntelligent,AlexandropoulosPhase}, i.e., the total of $N_R$ reflecting elements are partitioned into $R$ reflecting groups, each of which consists of $N_G=N_R/R$ adjacent elements sharing a common reflection coefficient.\footnote{In practice, the grouping should be implemented according to the statistical characteristics of channels. The value of $N_G$ can be obtained by trial and error for achieving a good trade-off between complexity, SE, and error performance. Also, the number of reflecting elements in each group is not necessary to be equal. Although the grouping that merges $N_G$ independent channels into a combined channel results in a loss of some diversity gains, the numbers of reflection coefficients and channel links that are required to be adjusted and estimated respectively are reduced from $N_R$ to $R$, which significantly lowers the complexity of RIS configuration and channel estimation.} Accordingly, the reflection coefficient for the $r$-th group of the RIS is expressed as $\phi_r=a_r\exp(j\theta_r)$, where $a_r$ is the amplitude coefficient and $\theta_r$ is the phase shift with $r=1,\ldots,R$. For all $r$, we set $a_r$ to unity and take $\theta_r \in [0,2\pi)$ from a discrete set $\Theta$.

The direct channel link from the transmitter to the receiver is described by the channel impulse response $\mathbf{g}_0=[g_0(1),\ldots,g_0(L_0)]^T$, where $L_0$ is the number of channel taps for $\mathbf{g}_0$. The combined channels from the transmitter to the $r$-th reflecting group and from the $r$-th reflecting group to the receiver are denoted by $\mathbf{h}_r=[h_r(1),\ldots,h_r(L_{r0})]^T$ with $L_{r0}$ channel taps and $\mathbf{f}_r=[f_r(1),\ldots,f_r(L_{r1})]^T$ with $L_{r1}$ channel taps, respectively. The equivalent channel of the cascaded transmitter-RIS-receiver link associated with the $r$-th reflecting group is given by $\mathbf{g}_r=[g_r(1),\ldots,g_r(L_r)]^T=\mathbf{h}_r \star \mathbf{f}_r$ with $L_r = L_{r0}+L_{r1}-1$. In this paper, the amplitude of $g_{\bar{r}}(l_{\bar{r}})$ is modeled as a Nakagami-$m$ distribution with fading parameter $m_{\bar{r}}(l_{\bar{r}})\in [1,2,\ldots,\infty]$ and spreading parameter $\Omega_{\bar{r}}(l_{\bar{r}}) >0$, i.e.,
\begin{align}
	p_{|g_{\bar{r}}(l_{\bar{r}})|}(x) = \frac{2m_{\bar{r}}(l_{\bar{r}})^{m_{\bar{r}}(l_{\bar{r}})}x^{2m_{\bar{r}}(l_{\bar{r}})-1}}{\Omega_{\bar{r}}(l_{\bar{r}})^{m_{\bar{r}}(l_{\bar{r}})}\Gamma(m_{\bar{r}}(l_{\bar{r}}))}\exp\left(-\frac{m_{\bar{r}}(l_{\bar{r}})x^2}{\Omega_{\bar{r}}(l_{\bar{r}})}\right),\quad x>0
\end{align}
where $\bar{r}=0,1,\ldots,R$ and $l_{\bar{r}}=1,\ldots,L_{\bar{r}}$. Note that the phase of $g_{\bar{r}}(l_{\bar{r}})$ is not uniformly distributed \cite{DashCoherent,MallikA}; its PDF is given in \cite{MallikA} and assumed to be the same for all $\bar{r}$ and $l_{\bar{r}}$. Here, the channels associated with different reflecting groups are approximately considered as independent and identically distributed fading. We note that the practical case of correlated Nakagami-$m$ fading \cite{AlexandropoulosSwitch} will be considered in future work.

\subsection{Transmitter Design}
For each transmission, the CPSC transmitter modulates $b = N\log_2(M)$ information bits into a symbol vector $\mathbf{x}=[x(1),\ldots,x(N)]^T$, where $x(n)$ is a normalized $M$-ary PSK symbol drawn from the constellation $\mathcal{X}$ for $n=1,\ldots,N$. After adding an $L$-length CP to the beginning of $\mathbf{x}$, the transmitted data vector can be given by
\begin{align}
	\mathbf{x}_{CP} = [x(N-L+1), \ldots, x(N), x(1),\ldots,x(N)]^T,
\end{align}
where $L\geq \max\{L_0,\ldots,L_R\}$.\footnote{The values of $L_0,\ldots,L_R$ are assumed to be known by using channel sounding schemes \cite{KimPerformance}.} At this point, the transmitter derives the values of $\theta_r$ for all $r$ based on $\mathbf{x}_{CP}$ and convey them to the RIS controller that further adjusts the $R$ reflecting groups. It is worth noting that since $\mathbf{x}_{CP}$ is a PSK-modulated vector, cyclically delayed versions of $\mathbf{x}_{CP}$ can be constructed by performing phase shifts on $\mathbf{x}_{CP}$. To achieve CDD, $\theta_r$ is required to vary along with the transmission of $\mathbf{x}_{CP}$ and thus denoted by the vector $\boldsymbol{\theta}_r$ in the following. Moreover, by applying $\boldsymbol{\theta}_r$, the signal reflected from the $r$-th group of the RIS is expected to be a cyclically delayed version of $\mathbf{x}$, appended with an additional $L$-length CP, i.e.,
\begin{align}
	\mathbf{x}_{CP}(\Delta_r) &= \Big[\underbrace{x(N-\Delta_r-L+1),\ldots,x(N-\Delta_r)}_{\text{CP}}, \nonumber \\
	&\hspace{+0.7cm}\underbrace{x(N-\Delta_r+1),\ldots,x(N),x(1),\ldots,x(N-\Delta_r)}_{\mathbf{x}(\Delta_r)}\Big]^T,
\end{align}
where $\mathbf{x}(\Delta_r)$ is the cyclically delayed version of $\mathbf{x}$ with the cyclic delay $\Delta_r$ for the $r$-th reflecting group. In CPSC-RIS, we let $\Delta_r=r\Delta$, where $L \leq \Delta \leq \lfloor{N/(R+1)}\rfloor$. Here, the constraint $\Delta \geq L$ is required to avoid inter-symbol interference (ISI) caused by multipath channels, while the constraint $\Delta \leq \lfloor{N/(R+1)}\rfloor$ is required to avoid ISI caused by simultaneous CPSC transmissions from different reflecting groups considering the block length. From $\mathbf{x}_{CP} \odot \exp(j\boldsymbol{\theta}_r)=\mathbf{x}_{CP}(\Delta_r)$, we have
\begin{align}\label{theta}
	\boldsymbol{\theta}_r = \Big[\angle x(N-\Delta_r-L+1)-\angle x(N-L+1),\ldots,\angle x(N-\Delta_r)-\angle x(N), \nonumber \\
	\angle x(N-\Delta_r+1)- \angle x(1),\ldots,\angle x(N-\Delta_r) - \angle x(N)\Big]^T,\quad r=1,\ldots,R
\end{align}
and
\begin{align}
	\Theta = \left\{0,\frac{2\pi}{M},\cdots,\frac{2\pi(M-1)}{M}\right\}.
\end{align}

\begin{figure}[t]
	\centering
	\includegraphics[width=4.5in]{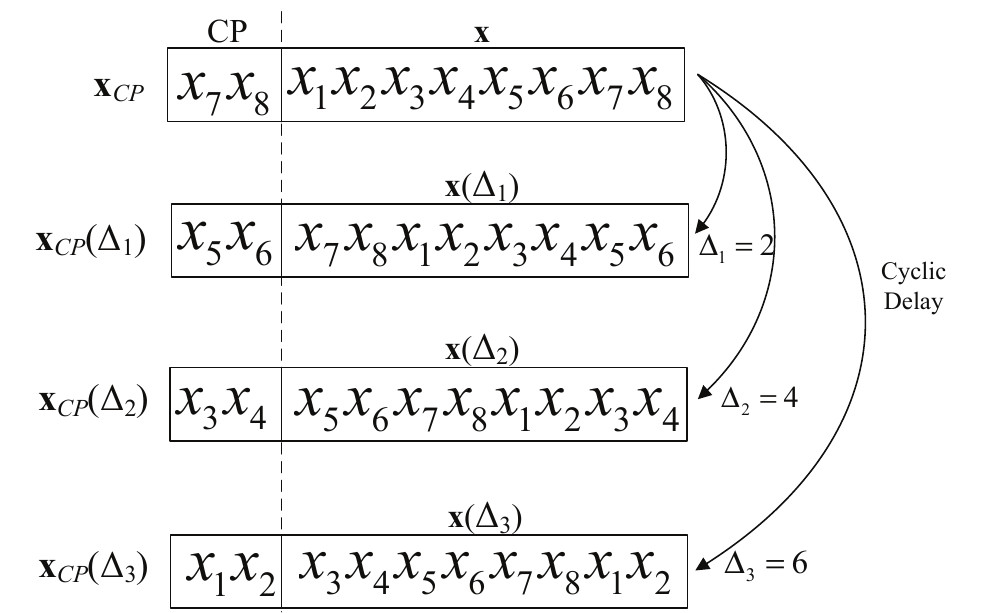}
	\caption{An example of $\mathbf{x}_{CP} $ and $\mathbf{x}_{CP}(\Delta_r),r= 1,\ldots,R$ for CPSC-RIS, where $N=8,R=3$, and $\Delta=2$. To save space, $x(n)$ is denoted by $x_n$ in the figure and later in Fig.~\ref{CDD_IM}.}
	\label{Fig_x}
\end{figure}
In Fig.~\ref{Fig_x}, we present an illustrative example of $\mathbf{x}_{CP} $ and $\mathbf{x}_{CP}(\Delta_r),r= 1,\ldots,R$, where $N=8,R=3,\Delta=2$, and $\mathbf{x}_{CP}(\Delta_r),r= 1,2,3$ are derived from $\mathbf{x}_{CP} $ via phase shifts.

At the receiver, after discarding the CP, the channel output is given by
\begin{align}\label{y_1}
	\mathbf{y} = \mathbf{G}_0 \mathbf{x} + \sum_{r=1}^{R}\mathbf{G}_r \mathbf{x}(\Delta_r) + \mathbf{w}= \sum_{\bar{r}=0}^{R}\mathbf{G}_{\bar{r}} \mathbf{x}(\Delta_{\bar{r}}) + \mathbf{w},
\end{align}
where $\Delta_0=0$, $\mathbf{w} \in \mathbb{C}^{N\times1}$ is the noise vector with the distribution $\mathcal{CN}(\mathbf{0},N_0\mathbf{I}_{N \times N})$, and $\mathbf{G}_{\bar{r}}$ is an $N \times N $ circulant channel matrix given by $\mathbf{G}_{\bar{r}}= \mathrm{cir}(\mathbf{g}_{\bar{r}}^0)$ with $\mathbf{g}_{\bar{r}}^0=[\mathbf{g}_{\bar{r}}^T,\mathbf{0}_{1 \times (N-L_{\bar{r}})}]^T \in \mathbb{C}^{N \times 1}$. The signal-to-noise ratio (SNR) in this paper is defined as $E_b/N_0$ with $E_b = (N + L)/b$ being the average transmitted energy per bit. We can rewrite (\ref{y_1}) as
\begin{align}\label{y_2}
	\mathbf{y} = \sum_{{\bar{r}}=0}^{R}\mathbf{X} \mathbf{g}_{\bar{r}}^0(\Delta_{\bar{r}}) + \mathbf{w}= \mathbf{X}\mathbf{g}_{eq} + \mathbf{w},
\end{align}
where $\mathbf{g}_{\bar{r}}^0(\Delta_{\bar{r}})$ is the cyclicly delayed version of $\mathbf{g}_{\bar{r}}^0$ with the cyclic delay $\Delta_{\bar{r}}$, $\mathbf{X}$ is an $N \times N $ circulant matrix represented by $\mathbf{X}=\mathrm{cir}(\mathbf{x})$,
and $\mathbf{g}_{eq}$ is given by
\begin{align}\label{g_eq}
	\mathbf{g}_{eq}=\sum_{{\bar{r}}=0}^{R}\mathbf{g}_{\bar{r}}^0(\Delta_{\bar{r}})=[\mathbf{g}_0^T,\mathbf{0}_{1 \times (\Delta-L_0)},\mathbf{g}_1^T,\mathbf{0}_{1 \times (\Delta-L_1)},\ldots,\mathbf{g}_R^T,\mathbf{0}_{1 \times (N-R\Delta -L_R)}]^T.
\end{align}
From (\ref{g_eq}), provided that the cyclic delay constraints are satisfied, increasing the value of $R$ enhances the channel impulse response (CIR). We also define $\mathbf{g}_{eq}'=[g_{eq}'(1),\ldots,g_{eq}'(L_s)]^T=[\mathbf{g}_0^T,\mathbf{g}_1^T,\ldots,\mathbf{g}_R^T]^T$, which consists of $L_s=\sum_{{\bar{r}}=0}^{R}L_{\bar{r}}$ non-zero entries of $\mathbf{g}_{eq}$. As seen from (\ref{y_1}) and (\ref{y_2}), the cyclic delays on $\mathbf{x}$ are converted into those on channels. By denoting $\mathbf{G}_{eq}=\mathrm{cir}(\mathbf{g}_{eq})$, (\ref{y_2}) can be further expressed as
\begin{align}\label{y_3}
	\mathbf{y} = \mathbf{G}_{eq}\mathbf{x} + \mathbf{w}.
\end{align}
It can be observed from (\ref{y_3}) that CPSC-RIS can be regarded as a conventional CPSC scheme with the enhanced CIR in (\ref{g_eq}). Employing CDD does not change the characteristics of the transmitted signal. Obviously, the SE of CPSC-RIS is given by
\begin{align}\label{SE1}
	\textsf{F}_{\text{CPSC-RIS}} = \frac{b}{N+L} =\frac{N\log_2(M)}{N+L} \quad \text{bps/Hz}.
\end{align}

\subsection{Receiver Design}
In this subsection, we successively design a pilot-aided channel estimator and various detectors for CPSC-RIS.

\subsubsection{Channel Estimator}
For channel estimation, a pilot vector $\mathbf{x}_p \in \mathbb{C}^{N\times1}$, instead of a data vector $\mathbf{x}$, is transmitted from the CPSC transmitter, and the RIS is configured similarly to (\ref{theta}). Based on (\ref{y_2}), the least-square estimation of $\mathbf{g}_{eq}$ is given by
\begin{align}\label{g_eq_hat}
	\hat{\mathbf{g}}_{eq} = \mathbf{X}_p^{-1}\mathbf{y} = \mathbf{g}_{eq} + \mathbf{g}_{e},
\end{align}
where $\mathbf{X}_p = \mathrm{cir}(\mathbf{x}_p)$ and $\mathbf{g}_{e} = \mathbf{X}_p^{-1}\mathbf{w}$ represents the vector of channel estimation errors. $\mathbf{g}_{e}$ is independent of  $\mathbf{g}_{eq}$ and follows the Gaussian distribution $\mathcal{CN}(\mathbf{0},\sigma_e^2\mathbf{I}_{N \times N})$. The MSE of channel estimation can be derived as
\begin{align}\label{MSE_Channel}
	\epsilon &=E\{\left\| {\mathbf{g}_{e}} \right\|^2\}  \nonumber \\
	&=\mathrm{Tr}\left\{\mathbf{X}_p^{-1}E\{\mathbf{w}\mathbf{w}^H\}\left(\mathbf{X}_p^{-1}\right)^H\right\} \nonumber \\
	&=N_0\mathrm{Tr}\left\{\left( \mathbf{X}_p^{H}\mathbf{X}_p\right)^{-1}\right\}.
\end{align}

	\begin{prop}
		To minimize the MSE, $\mathbf{X}_p$ should satisfy $\mathbf{X}_p^{H}\mathbf{X}_p=N\mathbf{I}_{N \times N}$.	
	\end{prop}
	
	\begin{proof}
		See Appendix A.
	\end{proof}

In particular, using Zadoff-Chu sequences \cite{ChuPolyphase} as $\mathbf{x}_p$ meets the requirement. For $N$ being an even number, we take $\mathbf{x}_p = [1,\exp({j\varpi \pi /N}),\exp( {j4\varpi \pi /N}), \ldots ,\exp({j\varpi \pi( {N - 1} )}^2/N)]^T$, where $\varpi$ is an arbitrary integer relatively prime to $N$. Obviously, in this case, the achieved MSE is given by $\epsilon = N_0$ and we have $\sigma_e^2 = N_0/ N$.
In Fig.~\ref{CE}, we verify the theoretical analysis of MSE by comparing it with simulation results, where $N=16,R=4,m_{\bar{r}}(l_{\bar{r}})=3$, and $L_{\bar{r}}=2$ for all $\bar{r}$ and $l_{\bar{r}}$ with an exponential-decay power-delay profile (PDP). It can be observed that the theoretical curve perfectly agrees with the simulation results for our proposed channel estimator.

\begin{figure}[t]
	\centering
	\includegraphics[width=4.0in]{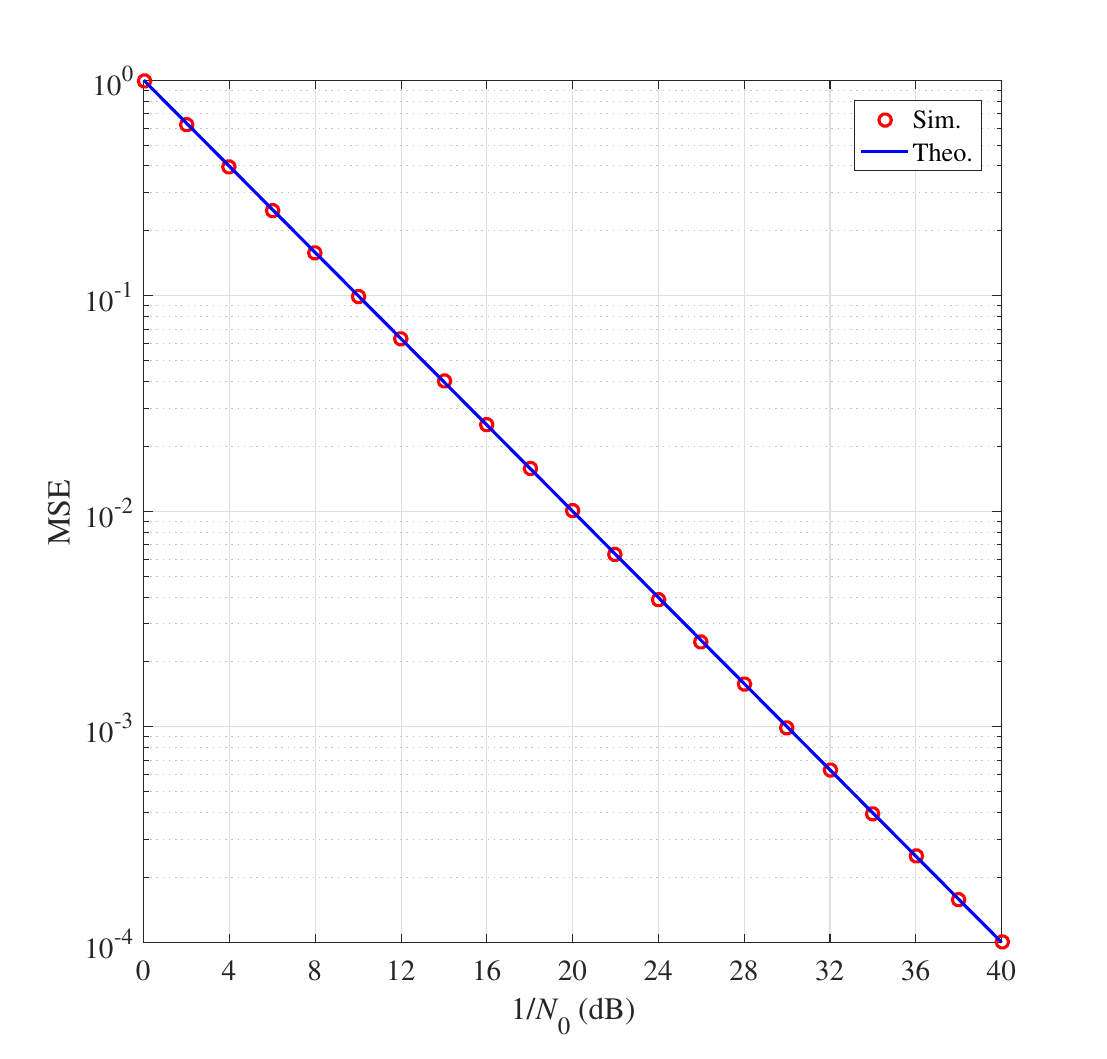}
	\caption{MSE versus $1/N_0$ for CPSC-RIS, where $N=16,R=4,m_{\bar{r}}(l_{\bar{r}})=3$, and $L_{\bar{r}}=2$ for all $\bar{r}$ and $l_{\bar{r}}$ with an exponentially decaying PDP.}
	\label{CE}
\end{figure}

As described above, the channel estimation of CPSC-RIS can be completed via a single block of pilots regardless of the value of $R$. This is much faster than existing frequency-domain channel estimation methods for RIS-aided broadband communications systems, such as the one in \cite{ZhengIntelligent} that requires $R+1$ transmission blocks to extract the CSI. In this sense, CPSC-RIS can be considered as a novel efficient framework of channel estimation in RIS-aided broadband communications. With CSI, various detectors can be developed to recover $\mathbf{x}$, which will be described in the following.
\subsubsection{ML Detector}
From (\ref{y_2}), the ML detector with the estimated CSI can be formulated as
	\begin{align}\label{ML}
		\hat{\mathbf{x} } =\arg \mathop {\min }\limits_{\bf{x}} {\left\| {{\bf{y}} -  { \mathrm{cir}(\bf{x})\hat{\mathbf{g}}_{eq}}} \right\|^2},
	\end{align}
	where $\hat{\mathbf{x} }=[\hat{x}(1),\ldots,\hat{x}(N)]^T$ is the estimate of $\mathbf{x}$. Since the ML detector makes a joint decision on the $N$ $M$-ary PSK symbols, the computational complexity in terms of complex multiplications is of order $\sim \mathcal{O}(M^N)$, which poses intolerable computational burden to the receiver. Thus, a low-complexity detector is highly recommended.

\subsubsection{ZF/MMSE Detectors}
Since $\mathbf{G}_{eq}$ is a circulant matrix, it has eigen decomposition $\mathbf{G}_{eq} = \mathbf{F}^H \mathbf{\Lambda }  \mathbf{F}$, where $\mathbf{F}$ is the unitary discrete Fourier transform (DFT) matrix with $\mathbf{F}^H\mathbf{F}={\bf{I}}_{N \times N}$, and $\mathbf{\Lambda }$ is the diagonal matrix whose diagonal elements are the $N$-point (non-unitary) DFT of $\mathbf{g}_{eq}$, i.e.,
\begin{align}\label{Lambda}
	\lambda (k) = \sum\limits_{n = 1}^N {{g_{eq}}\left( n \right)\exp \left( { - j\frac{{2\pi \left( {n - 1} \right)\left( {k - 1} \right)}}{N}} \right)},\quad k=1,\ldots,N.
\end{align}
Therefore, based on (\ref{y_3}), the frequency-domain received signal can be written as
\begin{align}\label{y_F}
	\mathbf{y}_F =\left[y_F(1),\ldots,y_F(N)\right]^T= \mathbf{F}\mathbf{y} = \mathbf{\Lambda } \mathbf{x}_F + \mathbf{w}_F,
\end{align}
where $\mathbf{x}_F$ and $\mathbf{w}_F$ are the frequency-domain counterparts of $\mathbf{x}$ and $\mathbf{w}$, respectively. For frequency-domain ZF/MMSE equalization, $\mathbf{y}_F$ is fed into $N$ single-tap equalizers in parallel, each of which is simply realized by one complex-valued multiplication, i.e.,
\begin{align}\label{x_F}
	\hat{x}_F(n) &= \phi(n) y_F(n), \quad n=1,\ldots,N
\end{align}
where $\hat{x}_F(n)$ is the output of the $n$-th equalizer, and $\phi(n)=\lambda (n)^*/(|\lambda (n)|^2 + cN_0)$ with $c = 1$ for MMSE equalizer and $c = 0$ for ZF equalizer.
Finally, we arrive at the estimated time-domain symbol vector via
\begin{align}\label{x_T}
	\hat{\mathbf{x}} = \mathbf{F}^H \hat{\mathbf{x}}_F,
\end{align}
where $\hat{\mathbf{x}}_F = [\hat{x}_F(1),\ldots,\hat{x}_F(N)]^T$. 

Obviously, by using single-tap equalization and existing efficient algorithms with the implementation of the DFT, the ZF/MMSE detectors achieve much lower computational complexity than the ML detector of (\ref{ML}). As a compromise, the ZF/MMSE detectors perform worse than the ML detector in terms of BER performance. Since CPSC-RIS can be regarded as a conventional CPSC scheme with enhanced CIR, the behaviors of ZF, ML, and MMSE detectors for CPSC-RIS are similar to those for conventional CPSC schemes. Particularly, the ML and MMSE detectors are able to harvest some diversity gains, while the ZF detector cannot extract any multipath diversity since the equalization severely amplifies the noise at frequencies in deep fades \cite{DevillersAbout}.

\section{Performance Analysis}
In this section, we concentrate on the performance analysis of the ML detector for CPSC-RIS. An upper bound on the BER of CPSC-RIS in the absence of channel estimation errors is provided after deriving the pairwise error probability (PEP) in closed-form.

Let us first study the conditional PEP with channel estimation errors, namely $\Pr(\mathbf{X} \to \hat{\mathbf{X}} | \hat{\mathbf{g}}_{eq})$, which is the probability of detecting $\mathbf{X}$ as $\hat{\mathbf{X}}$ conditioned on $\hat{\mathbf{g}}_{eq}$. From (\ref{y_2}), (\ref{y_3}), and (\ref{ML}), we have
\begin{align}\label{CPEP}
	\Pr\left(\mathbf{X} \to \hat{\mathbf{X}} | \hat{\mathbf{g}}_{eq}\right) = \Pr\left(\left\|\mathbf{y} -  \mathbf{X}\hat{\mathbf{g}}_{eq}  \right\|^2 > \left\|\mathbf{y} -  \hat{\mathbf{X}}\hat{\mathbf{g}}_{eq}  \right\|^2 \right).
\end{align}
Based on (\ref{g_eq_hat}), $\mathbf{y}$ can be expressed as
\begin{align}\label{y_4}
	\mathbf{y} = \mathbf{X}\left( \hat{\mathbf{g}}_{eq} - \mathbf{g}_e \right) + \mathbf{w} =\mathbf{X} \hat{\mathbf{g}}_{eq} + \bar{\mathbf{w}},
\end{align}
where $\bar{\mathbf{w}} = -\mathbf{X} \mathbf{g}_e + \mathbf{w}$. After putting (\ref{y_4}) into (\ref{CPEP}), we are led to
\begin{align}\label{CPEP_1}
	\Pr\left(\mathbf{X} \to \hat{\mathbf{X}} | \hat{\mathbf{g}}_{eq}\right) &= \Pr\left( - \left\| \left(\mathbf{X}- \hat{\mathbf{X}}\right)\hat{\mathbf{g}}_{eq} \right\|^2 - 2\Re\left\{\bar{\mathbf{w}}^H\left(\mathbf{X}- \hat{\mathbf{X}}\right)\hat{\mathbf{g}}_{eq}\right\} > 0\right)  \nonumber \\
	& = \Pr\left(V>0\right).
\end{align}
We observe that $V$ is Gaussian distributed with
\begin{align}
	E\left\lbrace V \right\rbrace &= -\left\| \left(\mathbf{X}- \hat{\mathbf{X}}\right)\hat{\mathbf{g}}_{eq}  \right\|^2, \nonumber \\
	Var\left\lbrace V \right\rbrace &= 2 N_0\left\| \left(\mathbf{X}- \hat{\mathbf{X}}\right)\hat{\mathbf{g}}_{eq}   \right\|^2 +
	2 \sigma_e^2 \left\| \mathbf{X}^H \left(\mathbf{X}- \hat{\mathbf{X}}\right)\hat{\mathbf{g}}_{eq}   \right\|^2. \nonumber
\end{align}
Therefore, (\ref{CPEP_1}) can be calculated as
\begin{align}\label{CPEP_2}
	\Pr\left(\mathbf{X} \to \hat{\mathbf{X}} | \hat{\mathbf{g}}_{eq}\right) = Q\left( \frac{ \left\| \left(\mathbf{X}- \hat{\mathbf{X}}\right)\hat{\mathbf{g}}_{eq}  \right\|^2}{\sqrt{2N_0 \left\| \left(\mathbf{X}- \hat{\mathbf{X}}\right)\hat{\mathbf{g}}_{eq}   \right\|^2 + 2\sigma_e^2 \left\| \mathbf{X}^H \left(\mathbf{X}- \hat{\mathbf{X}}\right)\hat{\mathbf{g}}_{eq}   \right\|^2}}\right).
\end{align}
In order to obtain the unconditional PEP, (\ref{CPEP_2}) should be further averaged over $\hat{\mathbf{g}}_{eq}$. However, it is very difficult to derive a closed-form expression. Here, we consider the case of perfect channel estimation, i.e., $\sigma_e^2=0$.

In the absence of channel estimation errors, (\ref{CPEP_2}) reduces to
\begin{align}\label{CPEP_3}
	\Pr\left(\mathbf{X} \to \hat{\mathbf{X}} |\mathbf{g}_{eq}\right) = Q\left( \sqrt{\frac{\left\| \left(\mathbf{X}- \hat{\mathbf{X}}\right)\mathbf{g}_{eq}  \right\|^2 }{2N_0}}\right).
\end{align}
The squared norm $\|(\mathbf{X}- \hat{\mathbf{X}})\mathbf{g}_{eq}\|^2$ can equivalently be developed as
\begin{align}\label{squared_norm}
	\left\| (\mathbf{X}- \hat{\mathbf{X}})_{[L_s]}\mathbf{g}_{eq}'  \right\|^2 &= \mathbf{g}_{eq}'^H \mathbf{A}\mathbf{g}_{eq}' \nonumber \\
	&=\mathbf{g}_{eq}'^H \mathbf{U}^H \mathbf{D}  \mathbf{U} \mathbf{g}_{eq}' \nonumber \\
	&=\sum\limits_{l = 1}^{L_s}d(l)|\tilde{g}_{eq}(l)|^2,
\end{align}
where $(\mathbf{X}- \hat{\mathbf{X}})_{[L_s]}$ comprises $L_s$ columns of $(\mathbf{X}- \hat{\mathbf{X}})$ corresponding to $L_s$ non-zero entries of $\mathbf{g}_{eq}$, $\mathbf{A}=(\mathbf{X}- \hat{\mathbf{X}})_{[L_s]}^H(\mathbf{X}- \hat{\mathbf{X}})_{[L_s]}$ is decomposed as $\mathbf{A}= \mathbf{U}^H\mathbf{D}\mathbf{U}$ with $\mathbf{U}$ being a unitary matrix and $\mathbf{D} = \mathrm{diag}\{d(1),\ldots,d(L_s)\}$, and $\tilde{\mathbf{g}}_{eq}=[\tilde{g}_{eq}(1),\ldots,\tilde{g}_{eq}(L_s)]^T = \mathbf{U}\mathbf{g}_{eq}'$. Hence, (\ref{CPEP_3}) can be rewritten as
\begin{align}\label{CPEP_4}
	\Pr\left(\mathbf{X} \to \hat{\mathbf{X}} |\mathbf{g}_{eq}\right) = Q\left( \sqrt{\frac{ \sum\limits_{l = 1}^{L_s}d(l)|\tilde{g}_{eq}(l)|^2 }{2N_0}}\right),
\end{align}
where $\tilde{g}_{eq}(l)=\mathbf{u}_l^T\mathbf{g}_{eq}'$ with $\mathbf{u}_l=[u_{l1},\ldots,u_{lL_s}]^T$ and $u_{ll'}$ being the $(l,l')$-th element of $\mathbf{U}$ for $l,l'=1,\ldots,L_s$. Next, the approximations in \cite{MallikA} are used to derive the moment generating function (MGF) of $|\tilde{g}_{eq}(l)|^2$. Specifically, the real and imaginary parts of a complex Nakagami-$m$ fading gain can be approximated as $\mathcal{N}(\mu_X,\Omega_s/2)$ and $\mathcal{N}(\mu_Y,\Omega_s/2)$ distributions, respectively, where $\mu_X=\sqrt{\sqrt{1-1/m}}\sqrt{\Omega}\cos\phi$, $\mu_Y=\sqrt{\sqrt{1-1/m}}\sqrt{\Omega}\sin\phi$, and $\Omega_s=\Omega(1-\sqrt{1-1/m})$ with $m$ being the fading parameter, $\Omega$ being the spreading parameter, and $\phi$ being the angle parameter of the complex Nakagami-$m$ distribution. Hence, by defining $\mu_X(l)$, $\mu_Y(l)$, and 
	$\Omega_s(l)$ as the corresponding parameters that are associated with $g_{eq}'(l)$, we have $\Re\{\tilde{g}_{eq}(l)\} \sim \mathcal{N}(\Re\{\mathbf{u}_l^T\bm{\mu}\}, \mathbf{u}_{2l}^T\bm{\Omega}_s/2)$ and $\Im\{\tilde{g}_{eq}(l)\} \sim \mathcal{N}(\Im\{\mathbf{u}_l^T\bm{\mu}\}, \mathbf{u}_{2l}^T\bm{\Omega}_s/2)$, where $\bm{\mu}=[\mu_X(1)+j\mu_Y(1),\ldots,\mu_X(L_s)+j\mu_Y(L_s)]^T$, $\mathbf{u}_{2l}=[|u_{l1}|^2,\ldots,|u_{lL_s}|^2]^T$, and $\bm{\Omega}_s=[\Omega_s(1),\ldots,\Omega_s(L_s)]^T$. Hence, the MGF of $|\tilde{g}_{eq}(l)|^2$ can be given by
	\begin{align}\label{MGF}
		\mathcal{M}_l(t)=\frac{1}{1-t\mathbf{u}_{2l}^T\bm{\Omega}_s}\exp\left(\frac{t|\mathbf{u}_l^T\bm{\mu}|^2}{1-t\mathbf{u}_{2l}^T\bm{\Omega}_s}\right).
	\end{align}
	It can be shown that $|\mathbf{u}_l^T\bm{\mu}|^2$ in (\ref{MGF}) is irrelevant to the value of $\phi$. By using the well-known approximation $Q(x)\approx 1/12 \cdot e^{-x^2/2}+ 1/4 \cdot e^{-2x^2/3}$ \cite{ChianiNew}, the unconditional PEP can be derived as
	\begin{align}\label{UPEP}
		\Pr\left(\mathbf{X} \to \hat{\mathbf{X}}\right) &\approx E\left\lbrace \frac{1}{12}\prod_{l = 1}^{L_s}\exp\left(-\frac{ d(l)|\tilde{g}_{eq}(l)|^2}{4N_0}\right) + \frac{1}{4}\prod_{l = 1}^{L_s}\exp\left(-\frac{ d(l)|\tilde{g}_{eq}(l)|^2}{3N_0}\right) \right\rbrace  \nonumber \\
		&= \frac{1}{12}\prod_{l = 1}^{L_s} \mathcal{M}_l\left(-\frac{ d(l)}{4N_0}\right)
		+ \frac{1}{4}\prod_{l = 1}^{L_s} \mathcal{M}_l\left(-\frac{ d(l)}{3N_0}\right).
\end{align}
Finally, according to the union bounding technique, the BER of CPSC-RIS can be upper bounded by
\begin{align}\label{Upper_Bound}
	P_{e} \leq \frac{1}{{b{2^b}}}\sum\limits_{\bf{X}} {\sum\limits_{\hat{\mathbf{ X}}} {\Pr\left( {{\bf{X}} \to \hat{\mathbf{ X}}} \right)\xi \left( {{\bf{X}},\hat{\mathbf{ X}}} \right)} },
\end{align}
where $\xi( {{\bf{X}},\hat{\mathbf{ X}}})$ is the number of erroneous bits when $\mathbf{X}$ is detected as $\hat{\mathbf{X}}$.

\textit{Remark 1:} The contribution of $L_s$ independent modified channel taps of $\tilde{\mathbf{g}}_{eq}$ to the unconditional PEP in (\ref{UPEP}) suggests that the highest possible diversity order is equal to $L_s$. However, this maximum diversity order is achieved if and only if $d(l)$ are non-null for all $l$, $\mathbf{X}$, and $\hat{\mathbf{X}}$. In general, the diversity order achieved by CPSC-RIS is given by $\min \mathrm{rank}(\mathbf{A})$. Since there exist some error events, such as
\begin{align}
	\mathbf{X}=
	\begin{bmatrix}
		1 &1  &\cdots   & 1 \\
		1 &1   &\cdots   & 1 \\
		\vdots & \vdots   &\cdots   & \vdots  \\
		1&1   &\cdots   & 1 \\
		1&1  &\cdots   & 1\\	
	\end{bmatrix}, \quad
	\hat{\mathbf{X}} =
	\begin{bmatrix}
		-1 &-1  &\cdots   & -1 \\
		-1 &-1   &\cdots   & -1 \\
		\vdots & \vdots   &\cdots   & \vdots  \\
		-1&-1   &\cdots   & -1 \\
		-1&-1  &\cdots   & -1\\	
	\end{bmatrix}, \nonumber
\end{align}
which make $\mathrm{rank}(\mathbf{A})=1$, the asymptotical (i.e., at infinite SNR) diversity order of CPSC-RIS is equal to one. Since this type of error events rarely occurs, simulation results in Section V will show that the multipath diversity can be extracted by CPSC-RIS for moderate BER values and reasonably large values of $N$.

On the other hand, compared with conventional CPSC systems without RIS and CDD, CPSC-RIS involves extra $R$ RIS's reflecting links with $\sum_{{\bar{r}}=1}^{R}L_{\bar{r}}$ channel taps that are collected through CDD. Therefore, from (\ref{UPEP}) with $L_s >> L_0$, the BER performance of CPSC-RIS is expected to be much better than that of conventional CPSC systems without RIS and CDD. 

\section{IM-Empowered CPSC-RIS System}
In this section, we extend CPSC-RIS to CPSC-RIS-IM, which improves the SE of CPSC-RIS by employing the concept of CDD-IM \cite{WenCyclic}. In CPSC-RIS-IM, the information bits are modulated into the cyclic delays at the RIS in addition to conventional constellation information of $N$ PSK symbols. Specifically, a total of $b$ information bits are divided into two parts. The first part, consisting of $b_1=\lfloor \log_2(R!)\rfloor$ bits, is used for permuting the cyclic delays at the RIS. Unlike the CPSC-RIS scheme in which the cyclic delay for the $r$-th reflecting group is fixed to $\Delta_r=r\Delta$, the cyclic delay for the $r$-th reflecting group in CPSC-RIS-IM is given by
\begin{align}\label{Delta}
	\Delta_r=k_r\Delta,
\end{align}
where $r=1,\ldots, R$ and $\mathbf{k}=[k_1,k_2,\ldots,k_R]^T$ is a full permutation of $\{1, 2,\ldots,R\}$ determined by the $b_1$ bits. The mapping from $b_{1}$ bits to $\mathbf{k}$ can be realized by either a look-up table or the permutation method \cite{WenMultiple}. Table \ref{Table_1} presents an example of a mapping table between the $b_{1}$ bits and $\mathbf{k}$ for $R=3$ and $b_1=2$, where the last two permutations are unused.

\begin{table}[!t]
	\caption{An example of mapping table between $p_{1}$ bits and $\mathbf{k}$ for $R=3$.}
	\label{Table_1}
	\centering
	\begin{tabular}{|c||c|c|c|c|c|c|}
		\hline
		$b_1$ bits  & [0 0] & [0 1] & [1 0] & [1 1] & -- & --\\
		\hline
		\hline
		$\mathbf{k}$  & $[1,2,3]^T$ &$[2,1,3]^T$ &$[1,3,2]^T$ & $[2,3,1]^T$ &$[3,2,1]^T$ &$[3,1,2]^T$\\
		\hline	
	\end{tabular}
\end{table}

The second part, comprised of $b_2= N\log_2(M)$ bits, is mapped to $N$ PSK symbols $\mathbf{x} = [x(1),\ldots,x(N)]^T$ via the regular $M$-PSK constellation $\mathcal{X}$ that is also used in CPSC-RIS. Therefore, the SE of CPSC-RIS-IM is given by
	\begin{align}\label{SE3}
		\textsf{F}_{\text{CPSC-RIS-IM}}=\frac{N\log_2(M)+\lfloor\log_2(R!)\rfloor}{N+L} \quad \text{bps/Hz},
	\end{align}
	which shows that CPSC-RIS-IM can transmit $\lfloor\log_2(R!)\rfloor$ more bits per block transmission than CPSC-RIS. Obviously, increasing the value of $R$ leads to higher SE. However, a larger value of $R$ also incurs higher complexity of implementing the mapping between the $b_{1}$ bits and $\mathbf{k}$ at the transceiver, cyclic delays at the RIS, and signal detection at the receiver.
\begin{figure}[t]
	\centering
	\includegraphics[width=5.5in]{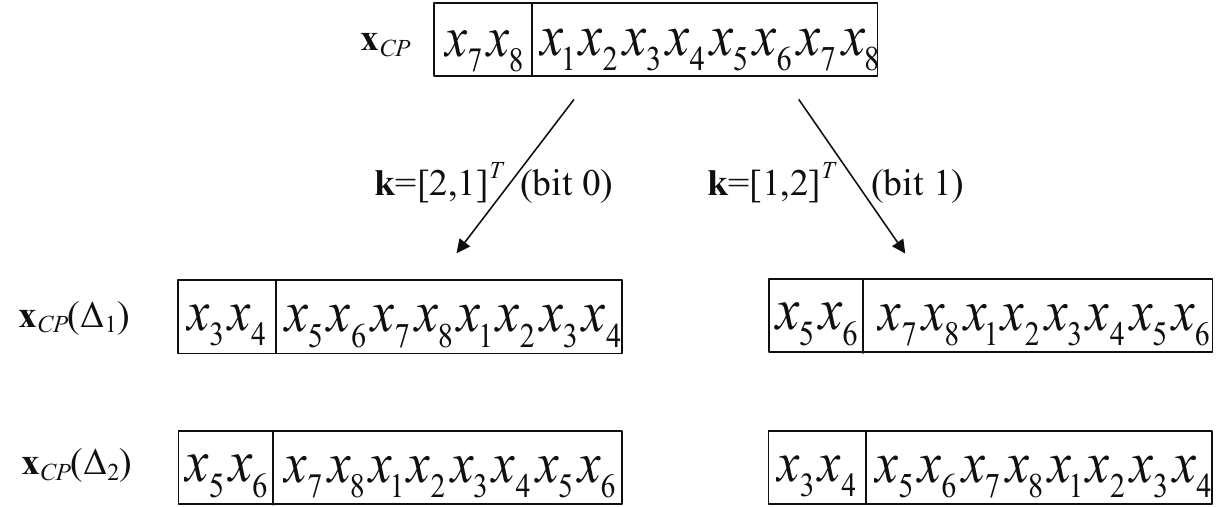}
	\caption{Examples of $\mathbf{x}_{CP} $ and $\mathbf{x}_{CP}(\Delta_r),r= 1,\ldots,R$ for CPSC-RIS-IM with $\mathbf{k}=[2,1]^T$ and $\mathbf{k}=[1,2]^T$, where $N=8, R=2$, and $\Delta=2$.}
	\label{CDD_IM}
\end{figure}

In Fig.~\ref{CDD_IM}, we present an example of $\mathbf{x}_{CP} $ and $\mathbf{x}_{CP}(\Delta_r),r= 1,\ldots,R$ for CPSC-RIS-IM, where $N=8, R=2, \Delta=2$, and both two possible realizations of $\mathbf{k}$ are considered. It should be noted that there is ambiguity at the receiver to detect $\mathbf{x}$ and $\mathbf{k}$. We exemplify this point with Fig.~\ref{CDD_IM} and letting $x(1) = x(2)=\cdots =x(8)$. Obviously, in this case, both $\mathbf{x}_{CP}(\Delta_1)$ and $\mathbf{x}_{CP}(\Delta_2)$ for $\mathbf{k}=[2,1]^T$ are completely the same as those for $\mathbf{k}=[1,2]^T$. Hence, we take the first symbol of $\mathbf{x}$, namely $x(1)$, as an anchor point to avoid ambiguity in signal detection. Specifically, instead of $\mathbf{x}$, we transmit $\tilde{\mathbf{x}}=\mathbf{x} \odot [\exp(j\pi/M),\mathbf{1}_{1 \times  (N-1)}]^T$ at the CPSC transmitter. In this manner, the Euclidean distance between $x(1)$ and any other symbols in $\tilde{\mathbf{x}}$ is maximized.

Similar to (\ref{y_2}), the received signal for CPSC-RIS-IM can be expressed as
\begin{align}\label{y_IM2}
	\mathbf{y} = \tilde{\mathbf{X}}\mathbf{g}_{eq,\mathbf{k}} + \mathbf{w},
\end{align}
where the notations are defined in (\ref{y_2}) except that $\tilde{\mathbf{X}} = \mathrm{cir}(\tilde{\mathbf{x}})$ and $\mathbf{g}_{eq,\mathbf{k}}$ is revised as
\begin{align}
	\mathbf{g}_{eq,\mathbf{k}}=\sum_{{\bar{r}}=0}^{R}\mathbf{g}_{\bar{r}}^0(k_r\Delta),
\end{align}
where $k_0=0$. Based on (\ref{y_IM2}), ML and low-complexity detectors can be designed for CPSC-RIS-IM, which are discussed in the following.

\subsubsection{ML Detector}
The optimal ML detector makes a joint decision on $\mathbf{k}$ and $\mathbf{x}$ by searching all possible combinations
of them, namely
\begin{align}\label{ML_IM2}
	\left( \hat{\mathbf{x} }, \hat{\mathbf{k} } \right) =\arg \mathop {\min }\limits_{\mathbf{x},\mathbf{k}} {\left\| {{\bf{y}} -  \tilde{\mathbf{X}}{\mathbf{g}_{eq,\mathbf{k}}}} \right\|^2},
\end{align}
where $\hat{\mathbf{x} }$ and $\hat{\mathbf{k}}$ are the estimates of $\mathbf{x}$ and $\mathbf{k}$, respectively.
Then, the corresponding $b$ information bits can be readily recovered from $\hat{\mathbf{x} }$ and $\hat{\mathbf{k}}$. Obviously, the ML detector should search all $R!M^N$ possible combinations of $\mathbf{x}$ and $\mathbf{k}$ to find out the one that minimizes the ML metric, which results in the computational complexity in terms of complex multiplications of order $\sim \mathcal{O}(R!M^N)$ per block transmission. To lower the computational complexity, we next develop a low-complexity detector.

\subsubsection{Low-Complexity Detector}
For a realization of $\mathbf{k}$, denoted by $\mathbf{k}_{i},i=1,\ldots,2^{b_1}$, we can employ the ZF/MMSE detectors similar to (\ref{y_F})-(\ref{x_T}) for obtaining the estimate of $\tilde{\mathbf{x}}$, denoted by $\hat{\tilde{\mathbf{x}}}_{\mathbf{k}_{i}}$. Further, given $\mathbf{k}_{i}$, $\mathbf{x}$ is estimated as $\hat{\mathbf{x}}_{\mathbf{k}_{i}}=[\hat{x}_{\mathbf{k}_{i}}(1),\ldots,\hat{x}_{\mathbf{k}_{i}}(N)]^T=\hat{\tilde{\mathbf{x}}}_{\mathbf{k}_{i}} \odot [\exp(-j\pi/M),\mathbf{1}_{1 \times  (N-1)}]^T$. Then, each symbol in $\hat{\mathbf{x}}_{\mathbf{k}_{i}}$ is demodulated independently via $M$-PSK demodulation, namely 
\begin{align}\label{ML_IM2_s}
	\hat{s}_{\mathbf{k}_{i}}(n) =\arg \mathop {\min }\limits_{s \in \mathcal{X}} {\left\|\hat{x}_{\mathbf{k}_{i}}(n) - s \right\|^2}, \quad n =1,\ldots,N
\end{align}
where $\hat{s}_{\mathbf{k}_{i}}(n)$ represents the hard decision on $\hat{x}_{\mathbf{k}_{i}}(n)$. Based on $\hat{\mathbf{s}}_{\mathbf{k}_{i}} = [\hat{s}_{\mathbf{k}_{i}}(1),\ldots,\hat{s}_{\mathbf{k}_{i}}(N)]^T$, the metric that $\mathbf{k}_{i}$ is used at the transmitter can be given by
\begin{align}\label{metric}
	T(\mathbf{k}_{i}) = {\left\| {{\bf{y}} -  \hat{\tilde{\mathbf{S}}}_{\mathbf{k}_{i}}{\mathbf{g}_{eq,\mathbf{k}_{i}}}} \right\|^2},
\end{align}
where $\hat{\tilde{\mathbf{S}}}_{\mathbf{k}_{i}} = \mathrm{cir}(\hat{\tilde{\mathbf{s}}}_{\mathbf{k}_{i}})$ with $\hat{\tilde{\mathbf{s}}}_{\mathbf{k}_{i}} = \hat{\mathbf{s}}_{\mathbf{k}_{i}} \odot [\exp(j\pi/M),\mathbf{1}_{1 \times  (N-1)}]^T$. Finally, the estimates of $\mathbf{k}$ and $\mathbf{x}$ can be expressed as
\begin{align}
	\hat{\mathbf{k}} = \min\limits_{\mathbf{k}_i} T(\mathbf{k}_{i}),
\end{align}
and $\hat{\mathbf{x}} = \hat{\mathbf{s}}_{\hat{\mathbf{k}}}$, respectively.

\subsubsection{Performance Analysis}
	Here, we analyze the BER performance of ML detection for CPSC-RIS-IM. Let us first define a permutation matrix $\mathbf{P}_{\mathbf{k}}$, which is made up of $N$ columns of $\mathbf{I}_{N\times N}$ and satisfies $\mathbf{P}_{\mathbf{k}}\mathbf{g}_{eq}=\mathbf{g}_{eq,\mathbf{k}}$. Then, (\ref{y_IM2}) can be rewritten as
	\begin{align}
		\mathbf{y} = \tilde{\mathbf{X}}\mathbf{P}_{\mathbf{k}}\mathbf{g}_{eq} + \mathbf{w}.
	\end{align}
	
	In the case of perfect channel estimation, similar to (\ref{CPEP_3}), the conditional PEP can be written as 
	\begin{align}\label{CPEP_IM}
		\Pr\left((\tilde{\mathbf{X}},{\mathbf{k}}) \to (\hat{\tilde{\mathbf{X}}},{\hat{{\mathbf{k}}}}) |{\mathbf{g}}_{eq}\right) &= \Pr\left(\left\|\mathbf{y} -  \tilde{\mathbf{X}}\mathbf{P}_{\mathbf{k}}{\mathbf{g}}_{eq}  \right\|^2 > \left\|\mathbf{y} -  \hat{\tilde{\mathbf{X}}}\mathbf{P}_{\hat{{\mathbf{k}}}}{\mathbf{g}}_{eq}  \right\|^2 \right) \nonumber \\
		&=Q\left( \sqrt{\frac{\left\| \left(\tilde{\mathbf{X}}\mathbf{P}_{\mathbf{k}}- \hat{\tilde{\mathbf{X}}}\mathbf{P}_{\hat{{\mathbf{k}}}}\right)\mathbf{g}_{eq}  \right\|^2 }{2N_0}}\right).
	\end{align}
	The unconditional PEP, namely $\Pr((\tilde{\mathbf{X}},{\mathbf{k}}) \to (\hat{\tilde{\mathbf{X}}},{\hat{{\mathbf{k}}}}))$, can be obtained by following the methods in (\ref{squared_norm})-(\ref{UPEP}), which is omitted to avoid redundancy. Finally, an upper bound on the BER of CPSC-RIS-IM is given by 
	\begin{align}
		P_{e} \leq \frac{1}{{b{2^b}}}\sum\limits_{\tilde{\mathbf{X}},{\mathbf{k}}} {\sum\limits_{\hat{\tilde{\mathbf{X}}},\hat{{\mathbf{k}}}}} {\Pr\left((\tilde{\mathbf{X}},{\mathbf{k}}) \to (\hat{\tilde{\mathbf{X}}},{\hat{{\mathbf{k}}}})\right)\xi \left( (\tilde{\mathbf{X}},{\mathbf{k}}),(\hat{\tilde{\mathbf{X}}},\hat{{\mathbf{k}}}) \right)},
	\end{align}
	where $\xi ( (\tilde{\mathbf{X}},{\mathbf{k}}),(\hat{\tilde{\mathbf{X}}},\hat{{\mathbf{k}}}))$ is the number of erroneous bits when $(\tilde{\mathbf{X}},{\mathbf{k}})$ is detected as $(\hat{\tilde{\mathbf{X}}},\hat{{\mathbf{k}}})$.

\textit{Remark 2:} In CPSC-RIS and CPSC-RIS-IM, the amplitude coefficients of all reflecting groups are fixed to 1, i.e., $a_r=1$ for $r=1,\ldots,R$. Actually, a reflection-type power amplifier can be deployed for each RIS element, such that $a_r$ can be drawn from a discrete set, say $\{a^{1},\ldots,a^{t}\}$ with all entries equal to or greater than 1 \cite{ZhangActive}. In this manner, the SE of CPSC-RIS and CPSC-RIS-IM can be further improved by encoding partial information (up to $R\log_2(t)$ bits) into the amplitude coefficients, and a multi-ring PSK constellation is observed at the receiver. Obviously, there is a trade-off between the SE and error performance.

\section{Simulation Results}
In this section, we conduct Monte Carlo simulations to evaluate the uncoded BER performance of CPSC-RIS(-IM) by taking CPSC and OFDM-RIS \cite{ZhengIntelligent} as benchmarks. In OFDM-RIS, the RIS is configured to maximize the channel gain according to the method in \cite{ZhengIntelligent} by assuming perfect CSI. In all simulations, we plot the BER versus SNR = $E_b/N_0$. We use $N_G=8$, while increasing it is expected to improve the performance for a given value of $R$ since the signal power reflected from each group is enhanced. The distances between the transmitter and the receiver, between the transmitter and the RIS, and between the RIS and the receiver are $D_0=50$ m, $D_1=5$ m, and $D_2=50$ m, respectively. The large-scale path loss of the transmitter-to-receiver, transmitter-to-RIS, and RIS-to-receiver links are given by $D_0^{-2.5}$, $D_1^{-2}$, and $D_2^{-2}$, respectively. All wireless channels are modeled by the exponentially	decaying PDP with  the decaying factor of unity, where each tap is generated according to the Nakagami-$m$ distribution. For simplicity, we assume that $L_0=\cdots = L_R=2$, $m_0(1) =\cdots =m_0(L_0)=m_1(1) =\cdots = m_1(L_1)=\cdots=m_R(L_R)=2$, and $\Delta=L$. Each BER point is obtained by averaging over at least $10^5$ channel realizations.

\begin{figure}[t]
	\centering
	\includegraphics[width=4.5in]{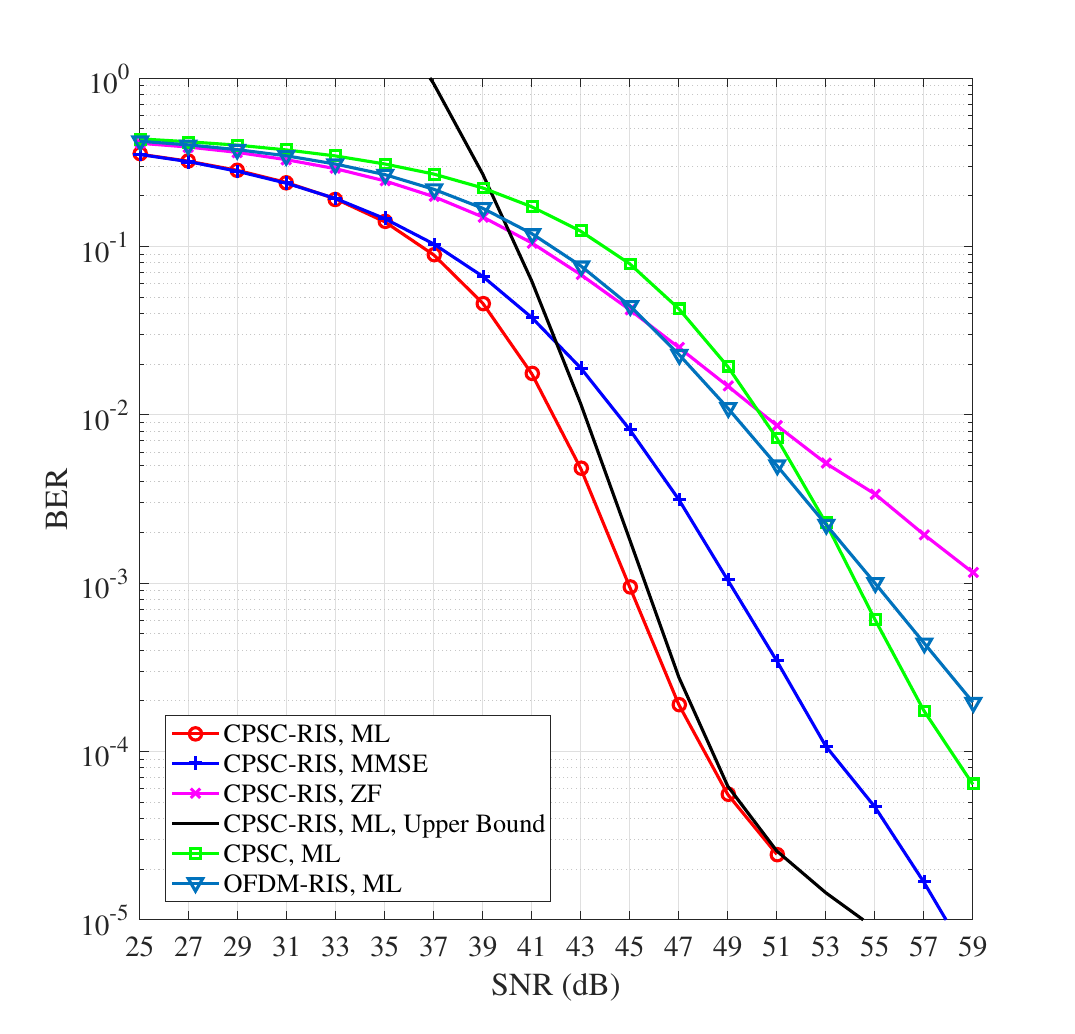}
	\caption{Performance comparison among CPSC-RIS, CPSC, and OFDM-RIS, where $N=8,M=2,R=2$, and perfect CSI is assumed at the receiver.}
	\label{Fig1}
\end{figure}

\begin{figure}[t]
\centering
\includegraphics[width=4.5in]{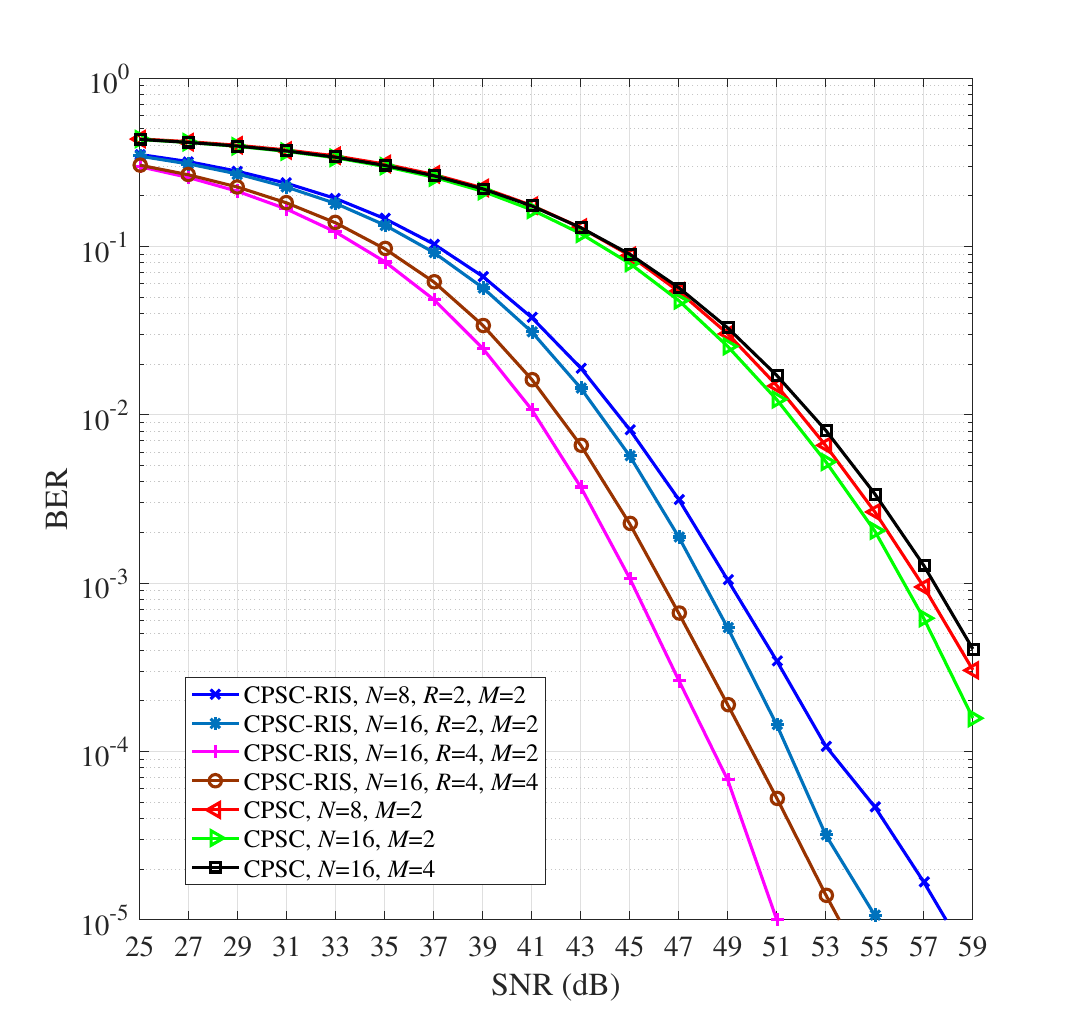}
\caption{Performance comparison between CPSC-RIS and CPSC, where $N=8,16,M=2,4,R=2,4$, and perfect CSI is assumed at the receiver.}
\label{Fig2}
\end{figure}
Fig.~\ref{Fig1} depicts the performance comparison among CPSC-RIS, CPSC, and OFDM-RIS, where $N=8,M=2,R=2$, and perfect CSI is assumed at the receiver. CPSC-RIS employs the ML, MMSE, and ZF detectors, while CPSC and OFDM-RIS use the ML detectors. To verify the analysis given in Section III, we also plot the BER upper bound (\ref{Upper_Bound}) as a theoretical bound for the ML detector of CPSC-RIS in Fig.~\ref{Fig1}. As seen from Fig.~\ref{Fig1}, for CPSC-RIS with the optimal ML detection, the BER upper bound agrees with the simulation results in the high SNR region. In particular, for CPSC-RIS, the ZF detector cannot harvest any diversity and performs the worst. In contrast, for CPSC-RIS, both the ML and MMSE detectors achieve some diversity gains for practical BER values, while they fails to extract diversity gains at infinite SNR. Moreover, CPSC-RIS even with low-complexity MMSE detector significantly outperforms conventional CPSC and OFDM-RIS with the optimal ML detectors throughout the considered SNR region. CPSC-RIS with the ML and MMSE detectors obtain approximately 10 dB and 5 dB SNR gains, respectively, over CPSC with the ML detector at a BER value of $10^{-4}$. Since the MMSE detector is able to achieve acceptable performance with low implementation complexity, the MMSE-based detectors are adopted for CPSC-RIS(-IM) and CPSC in the remaining simulations.

Fig.~\ref{Fig2} evaluates the BER performance of CPSC-RIS and CPSC, where $N=8,16,M=2,4,R=2,4$, and perfect CSI is assumed at the receiver. It can be observed from Fig.~\ref{Fig2} that all considered CPSC-RIS schemes perform better than all
CPSC schemes for all SNR values thanks to the extra $R$ reflecting links introduced by the RIS, and the gain is more prominent for a larger value of $R$. For $N=16$ and $M=2$, CPSC-RIS with $R=4$ obtains about 3 dB SNR gain over that with $R=2$, at a BER value of $10^{-4}$. Increasing the value of $N$ achieves SNR gains of 2 dB and 1 dB for CPSC-RIS and CPSC, respectively, at a BER value of $10^{-4}$. By contrast, a larger constellation size deteriorates the performance (about 2 dB SNR loss) for both CPSC-RIS and CPSC.

\begin{figure}[t]
	\centering
	\includegraphics[width=4.5in]{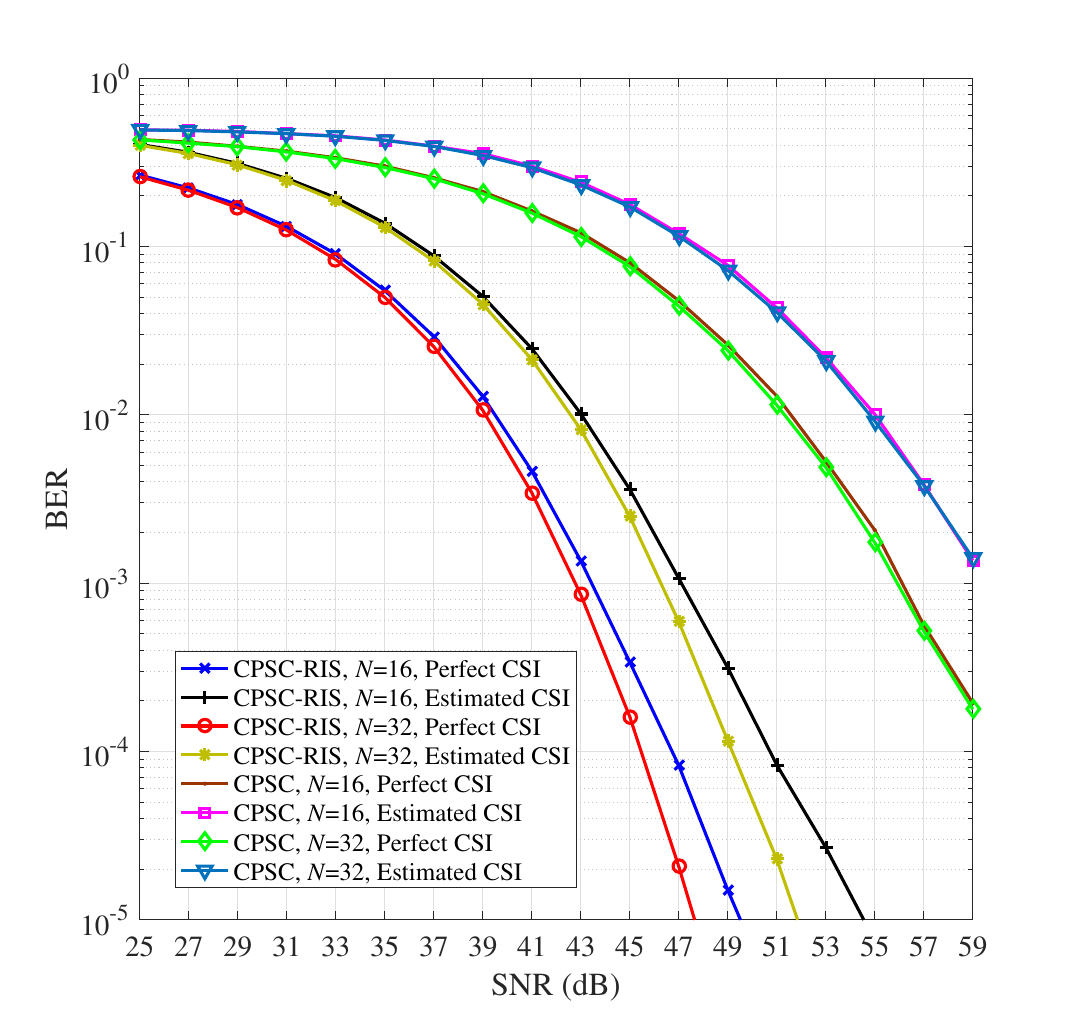}
	\caption{Performance comparison between CPSC-RIS and CPSC, where $N=16,32,M=2$, and $R=6$. The cases of perfect channel estimation and practical channel estimation in (\ref{g_eq_hat}) are considered.}
	\label{Fig3}
\end{figure}
Fig.~\ref{Fig3} presents the performance comparison between CPSC-RIS and CPSC with perfect and estimated CSI at the receiver, where $N=16,32,M=2$, and $R=6$. The estimated CSI is obtained via (\ref{g_eq_hat}). As expected, all schemes with estimated CSI perform worse than the corresponding schemes with perfect CSI, and approximately 4 dB SNR loss is incurred by the imperfect CSI. We observe from Fig.~\ref{Fig3} that, in both cases, increasing the value of $N$ from 16 to 32 leads to SNR gains for both CPSC-RIS and CPSC. Specifically, the SNR gain for CPSC-RIS is about 1 dB at a BER value of $10^{-4}$, while that for CPSC is minor. With both perfect and estimated CSI at the receiver, CPSC-RIS outperforms CPSC for all considered SNR values. Moreover, for $N=16$ and 32, CPSC-RIS with estimated CSI achieves about 9 dB and 10 dB SNR gains, respectively, over CPSC with perfect CSI, at a BER value of $10^{-3}$.

\begin{figure}[t]
	\centering
	\includegraphics[width=4.5in]{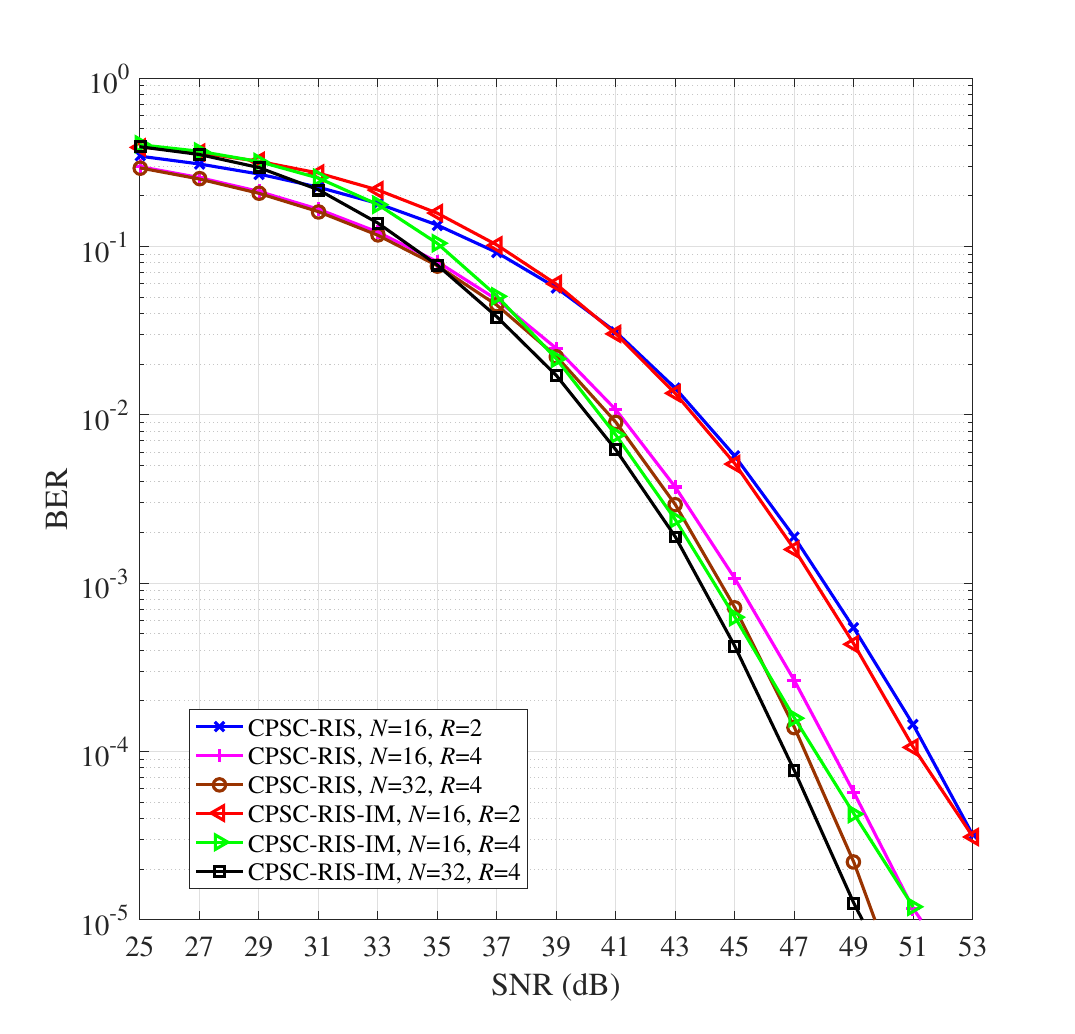}
	\caption{Performance comparison between CPSC-RIS and CPSC-RIS-IM, where $N=16,32,M=2,R=2,4$, and perfect CSI is assumed at the receiver.}
	\label{Fig4}
\end{figure}

In Fig.~\ref{Fig4}, we compare the BER performance of CPSC-RIS and CPSC-RIS-IM, where $N=16,32,M=2,R=2,4$, and perfect CSI is assumed at the receiver. As seen from Fig.~\ref{Fig4}, CPSC-RIS-IM outperforms CPSC-RIS at high SNR with the same parameter settings. Particularly, the performance gain achieved by CPSC-RIS-IM over CPSC-RIS for $R=4$ is larger than that for $R=2$. This can be explained as follows. The cyclic delay constraints in Section II are met for both $R=2$ and $R=4$. As seen from (\ref{g_eq}), increasing the value of $R$ from 2 to 4 enhances the equivalent CIR. Hence, the performance of both CPSC-RIS-IM and CPSC-RIS becomes better when $R$ increases from 2 to 4. Moreover, for CPSC-RIS-IM, a larger value of $R$ results in more transmitted IM bits, and the transmit power becomes higher for a given value of SNR=$E_b/N_0$. In addition, the transmission of IM bits itself does not consume power. Therefore, the performance improvement for CPSC-RIS-IM is more significant than that for CPSC-RIS when $R$ increases from 2 to 4. Also, increasing the value of $N$ enhances the performance of both CPSC-RIS and CPSC-RIS-IM. In particular, doubling the value of $R$ has a more significant impact than that of $N$ on both CPSC-RIS and CPSC-RIS-IM. Specifically, about 3 dB and 1 dB SNR gains can be achieved at a BER value of $10^{-4}$, by increasing the value of $R$ from 2 to 4 and the value of $N$ from 16 to 32, respectively.

\begin{figure}[t]
	\centering
	\includegraphics[width=4.5in]{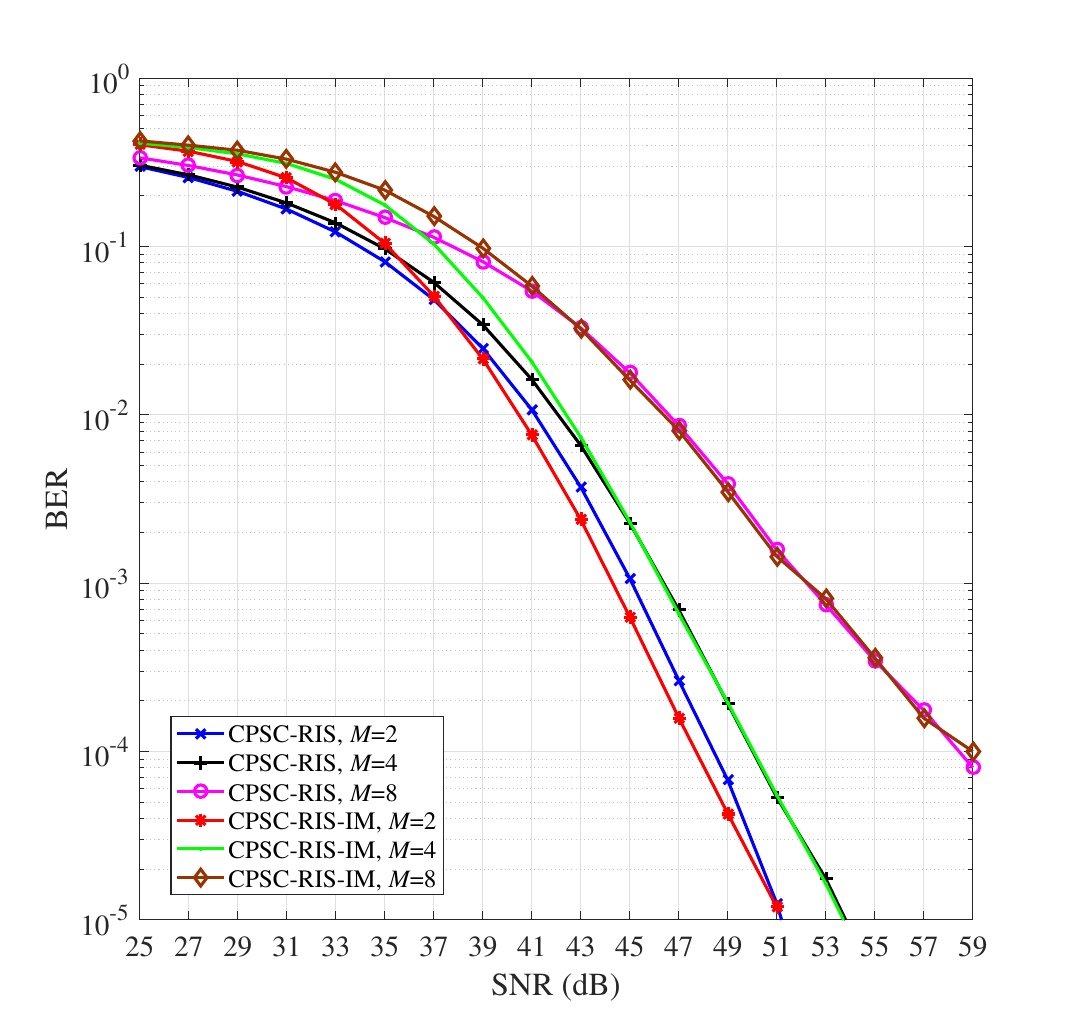}
	\caption{Performance comparison between CPSC-RIS and CPSC-RIS-IM, where $N=16,R=4,M=2,4,8$, and perfect CSI is assumed at the receiver.}
	\label{Fig5}
\end{figure}

Fig.~\ref{Fig5} illustrates the performance comparison between CPSC-RIS and CPSC-RIS-IM, where $N=16,R=4,M=2,4,8$, and perfect CSI is assumed at the receiver. As shown in Fig.~\ref{Fig5}, with $M=2$, CPSC-RIS-IM performs better than CPSC-RIS for practical BER values with SNR$>38$ dB. With increasing the value of $M$, the superiority of CPSC-RIS-IM over CPSC-RIS is lost. This is because that the proportion of the IM bits in CPSC-RIS-IM decreases with increasing the value of $M$. However, from Fig.~\ref{Fig2}, it is expected that CPSC-RIS-IM still achieves much better BER performance than conventional CPSC without RIS and CDD.

\section{Conclusion}
In this paper, we have proposed a CDD-enhancing CPSC transmission scheme for RIS-aided broadband wireless systems. The channel estimation, signal detection, and SE improvement have all been studied for CPSC-RIS. Specifically, a practical and efficient time-domain channel estimator has been developed for CPSC-RIS, which can be considered as a general channel estimator for RIS-aided broadband communications. The optimal ML and low-complexity ZF/MMSE detectors have been designed for CPSC-RIS, and the BER performance of the ML detector over frequency-selective Nakagami-$m$ fading channels has been analyzed with the theoretically derived upper bound. Further, CPSC-RIS has been extended to CPSC-RIS-IM for improving the SE by using the concept of IM. We conclude that the proposed two schemes can be considered as promising candidates for RIS-empowered broadband wireless communications.

	\begin{appendices}
		\section{Proof of Proposition 1}
		Since $\mathbf{X}_p$ is a circulant matrix, $\mathbf{X}_p^H\mathbf{X}_p$ is also a circulant matrix and can be expressed as $\mathbf{X}_p^H\mathbf{X}_p = \mathrm{cir}(\mathbf{x}_{pp})$, where $\mathbf{x}_{pp}=[x_{pp}(1), x_{pp}(2), \ldots, x_{pp}(N)]^T$ with $x_{pp}(n)=\mathbf{x}_p^H(n-1)\mathbf{x}_p$ and $\mathbf{x}_p(n-1)$ being the cyclically delayed version of $\mathbf{x}_p$ with the cyclic delay $n-1$ for $n=1,\ldots,N$. Obviously, we have $x_{pp}(1)=\mathbf{x}_p^H\mathbf{x}_p=N$.
		
		Further, as a circulant matrix, $\mathbf{X}_p^H\mathbf{X}_p$ can be decomposed as $\mathbf{X}_p^H\mathbf{X}_p = \mathbf{F}^{-1}\mathbf{D}\mathbf{F}$, where $\mathbf{F}$ is the unitary DFT matrix and $\mathbf{D}_p=\mathrm{diag}([d_p(1),\ldots, d_p(N)]^T)$ with 
		\begin{align}
			d_p\left( k \right) =\sum_{n=1}^N{x_{pp}\left( n \right) \exp \left( -j\frac{2\pi \left( n-1 \right) \left( k-1 \right)}{N} \right)}, \quad k=1,\dots , N,
		\end{align}
		and $\sum\nolimits_{k=1}^N{d_p\left( k \right)}=\mathrm{Tr}\left\{ \mathbf{X}_{p}^{H}\mathbf{X}_p \right\} =N^2$. From (\ref{MSE_Channel}), the MSE can be rewritten as
		\begin{align}\label{epsilon}
			\epsilon =N_0\mathrm{Tr}\left\{ \left( \mathbf{X}_{p}^{H}\mathbf{X}_p \right) ^{-1} \right\} =N_0\mathrm{Tr}\left\{ \mathbf{F}^{-1}\mathbf{D}_{p}^{-1}\mathbf{F} \right\} =N_0\mathrm{Tr}\left\{ \mathbf{D}_{p}^{-1} \right\} =N_0\sum_{k=1}^N{\frac{1}{d_p\left( k \right)}}.
		\end{align}
		From (\ref{epsilon}), under the condition that $\sum\nolimits_{k=1}^N{d_p\left( k \right)}=N^2$, to minimize the MSE, we should have $d_p(1)=\ldots=d_p(N)=N$. Finally, we conclude that $x_{pp}=[N, 0,\ldots, 0]^T$ and $
		\mathbf{X}_{p}^{H}\mathbf{X}_p=N\mathbf{I}_{N\times N}$ when the MSE is minimized, which completes the proof.
	\end{appendices}


\begin{thebibliography}{99}
	
	\bibitem{PancaldiSingle}	
	F. Pancaldi, G. Vitetta, R. Kalbasi, N. Al-Dhahir, M. Uysal, and H. Mheidat, ``Single-carrier frequency domain equalization,'' \textit{IEEE Signal Process. Mag.}, vol.~25, no.~5, pp.~37--56, Sep.~2008.
	
	\bibitem{FalconerFrequency}	
	D. Falconer, S. Ariyavisitakul, A. Benyamin-Seeyar, and B. Eidson, ``Frequency domain equalization for single-carrier broadband wireless systems,'' \textit{IEEE Commun. Mag.}, vol.~40, no.~4, pp.~58--66, Apr.~2002.
	
	\bibitem{DevillersAbout}
	B. Devillers, Jerome Louveaux, and Luc Vandendorpe. ``About the diversity in cyclic prefixed single-carrier systems,'' \textit{Phys. Commun.}, vol.~1, no.~4, pp.~266--276, 2008.
	
	\bibitem{KimPerformance1}	
	K. J. Kim and T. A. Tsiftsis, ``Performance analysis of QRD-based cyclically prefixed single-carrier transmissions with opportunistic scheduling,'' \textit{IEEE Trans. Veh. Tech.}, vol. 60, no. 1, pp. 328-333, Jan. 2011.
	
	\bibitem{KimDiversity}	
	K. J. Kim, H. Liu, Z. Ding, P. V. Orlik, and H. V. Poor, ``Diversity gain analysis of distributed CDD systems in non-identical fading channels,'' \textit{IEEE Trans. Commun.}, vol. 68, no. 11, pp. 7218-7231, Nov. 2020.
	
	
	
	\bibitem{DammannStandard}
	A. Dammann and S. Kaiser, ``Standard comformable antenna diversity techniques for OFDM and its application to the DVB-T system,'' in \textit{Proc. IEEE Global Commun. Conf.}, San Antonio, TX, USA, Nov. 2001, pp.~388--397.
	
	
	\bibitem{KwonCyclic}
	U.-K. Kwon and G.-H. Im, ``Cyclic delay diversity with frequency domain turbo equalization for uplink fast fading channels,'' \textit{IEEE Commun. Lett.}, vol.~13, no.~3, pp.~184--186, Mar.~2009.
	
	\bibitem{LiangDesign}
	Y. Liang, W. S. Leon, Y. Zeng, and C. Xu, ``Design of cyclic delay diversity for single carrier cyclic prefix (SCCP) transmissions with block-iterative GDFE (BI-GDFE) receiver,''  \textit{IEEE Trans. Wireless Commun.}, vol.~7, no.~2, pp.~677--684, Feb.~2008.
	
	
	\bibitem{KimPerformance}
	K. J. Kim, M. Di Renzo, H. Liu, P. V. Orlik, and H. V. Poor, ``Performance analysis of distributed single carrier systems with distributed cyclic delay diversity,'' \textit{IEEE Trans. Commun.}, vol.~65, no.~12, pp.~5514--5528, Dec.~2017.
	
	\bibitem{KimSecrecy}
	K. J. Kim, H. Liu, M. Wen, P. V. Orlik, and H. V. Poor, ``Secrecy performance analysis of distributed asynchronous cyclic delay diversity-based cooperative single carrier systems,'' \textit{IEEE Trans. Commun.}, vol. 68, no. 5, pp. 2680--2694, May~2020.
	
	
	
	
	
	
	
	\bibitem{IradukundaOn}
	N. Iradukunda, H. T. Nguyen, and W. Hwang, ``On cyclic delay diversity-based single-carrier scheme in spectrum sharing systems,'' \textit{IEEE Commun. Lett.}, vol.~23, no.~6, pp.~1069--1072, Jun.~2019.
	
	
	\bibitem{BasarWireless}
	E. Basar, M. Di Renzo, J. De Rosny, M. Debbah, M. -S. Alouini, and R. Zhang, ``Wireless communications through reconfigurable intelligent surfaces,'' \textit{IEEE Access}, vol.~7, pp.~116753--116773, 2019.
	
	\bibitem{RenzoSmart}
	M. Di Renzo, M. Debbah, D.-T. Phan-Huy, A. Zappone, M.-S. Alouini, C. Yuen, V. Sciancalepore, G. C. Alexandropoulos, J. Hoydis, H. Gacanin, J. de Rosny, A. Bounceu, G. Lerosey, and M. Fink, ``Smart radio environments empowered by AI reconfigurable meta-surfaces: An idea whose time has come,'' \textit{EURASIP J. Wireless Commun. Netw.}, vol. 129, pp.~1--20, May 2019. 
	
	\bibitem{JianReconfigurable}
	M. Jian, G. C. Alexandropoulos, E. Basar, C. Huang, R. Liu, Y. Liu, and C. Yuen, ``Reconfigurable intelligent surfaces for wireless communications: Overview of hardware designs, channel models, and estimation techniques,'' \textit{Intelligent and Converged Networks}, vol.~3, no.~1, pp.~1--32, Mar.~2022.
	
	
	\bibitem{StrinatiReconfigurable}
	E. Calvanese Strinati, G. C. Alexandropoulos, H. Wymeersch, B. Denis, V. Sciancalepore,
	R. D’Errico, A. Clemente, D.-T. Phan-Huy, E. De Carvalho, and P. Popovski, ``Reconfigurable, intelligent, and sustainable wireless environments for 6G smart connectivity,'' \textit{IEEE Commun. Mag.}, vol.~59, no.~10, pp.~99--105, Oct. 2021.
	
	
	
	\bibitem{HuangReconfigurable}
	C. Huang, A. Zappone, G. C. Alexandropoulos, M. Debbah, and C. Yuen, ``Reconfigurable intelligent surfaces for energy efficiency in wireless communication,'' \textit{IEEE Trans. Wireless Commun.}, vol.~18, no.~8, pp.~4157--4170, Aug.~2019.
	
	
	\bibitem{WuIntelligent}
	Q. Wu and R. Zhang, ``Intelligent reflecting surface enhanced wireless network via joint active and passive beamforming,'' \textit{IEEE Trans. Wireless Commun.}, vol. 18, no. 11, pp. 5394--5409, Nov.~2019.
	
	
	\bibitem{GuoWeighted}
	H. Guo, Y. Liang, J. Chen, and E. G. Larsson, ``Weighted sum-rate maximization for reconfigurable intelligent surface aided wireless networks,'' \textit{IEEE Trans. Wireless Commun.}, vol.~19, no.~5, pp.~3064--3076, May 2020.
	
	
	\bibitem{HongRobust}
	S. Hong, C. Pan, H. Ren, K. Wang, K. K. Chai, and A. Nallanathan, ``Robust transmission design for intelligent reflecting surface-aided secure communication systems with imperfect cascaded CSI,'' \textit{IEEE Trans. Wireless Commun.}, vol.~20, no.~4, pp.~2487--2501, Apr.~2021.
	
	\bibitem{DuReconfigurable}
	L. Du, C. Huang, W. Guo, J. Ma, X. Ma, and Y. Tang, ``Reconfigurable intelligent surfaces assisted secure multicast communications,'' \textit{IEEE Wireless Commun. Lett.}, vol. 9, no. 10, pp. 1673--1676, Oct. 2020.
	
	\bibitem{AlexandropoulosSafeguarding}
	G. C. Alexandropoulos, K. Katsanos, M. Wen, and D. B. da Costa, ``Safeguarding MIMO communications with reconfigurable metasurfaces and artificial noise,'' in \textit{Proc. IEEE Int. Conf. Commun.}, Montreal, Canada, Jun. 2021, pp.~1--6.
	
	
	\bibitem{YildirimHybrid}
	I. Yildirim, F. Kilinc, E. Basar, and G. C. Alexandropoulos, ``Hybrid RIS-empowered reflection and decode-and-forward relaying for coverage extension,'' \textit{IEEE Commun. Lett.}, vol.~25, no.~5, pp.~1692--1696, May 2021.
	
	
	\bibitem{YeJoint}
	J. Ye, S. Guo, and M.-S. Alouini, ``Joint reflecting and precoding designs for SER minimization in reconfigurable intelligent surfaces assisted MIMO systems,'' \textit{IEEE Trans. Wireless Commun.}, vol.~19, no.~8, pp.~5561--5574, Aug.~2020.
	
	\bibitem{AbuNear}
	Z. Abu-Shaban, K. Keykhosravi, M. F. Keskin, G. C. Alexandropoulos, G. Seco-Granados, and H. Wymeersch, ``Near-field localization with a reconfigurable intelligent surface acting as lens,'' in \textit{Proc. IEEE Int. Conf. Commun.}, Montreal, Canada, Jun. 2021, pp.~1--6.
	
	
	\bibitem{TangMIMO}
	W. Tang, J. Y. Dai, M. Z. Chen, K.-K. Wong, X. Li, X. Zhao, S. Jin, Q. Cheng, and T. J. Cui, ``MIMO transmission through reconfigurable intelligent surface: System design, analysis, and implementation,''  \textit{IEEE J. Sele. Areas Commun.}, vol.~38, no.~11, pp.~2683--2699, Nov.~2020.
	
	
	
	\bibitem{TangRealization}
	W. Tang, J. Y. Dai, M. Z. Chen, Y. Han, X. Li, C.-K. Wen, S. Jin, Q. Cheng, and T. J. Cui, ``Realization of reconfigurable intelligent surface-based Alamouti space–time transmission,'' in \textit{Proc. IEEE Int. Conf. Wireless Commun. Signal Process.}, Nanjing, China, Oct.~2020, pp. 904--909.
	
	
	
	\bibitem{BasarIndex}
	E. Basar, ``Index modulation techniques for 5G wireless networks,'' \textit{IEEE Commun. Mag.}, vol. 54, no. 7, pp.~168--175, Jul.~2016.
	
	
	\bibitem{CanbilenReconfigurable}
	A. E. Canbilen, E. Basar, and S. S. Ikki, ``Reconfigurable intelligent surface-assisted space shift keying,'' \textit{IEEE Wireless Commun. Lett.}, vol.~9, no.~9, pp.~1495--1499, Sep.~2020.
	
	\bibitem{LiSpace}
	Q. Li, M. Wen, S. Wang, G. C. Alexandropoulos, and Y. -C. Wu, ``Space shift keying with reconfigurable intelligent surfaces: Phase configuration designs and performance analysis,'' \textit{IEEE Open J. Commun. Soc.}, vol. 2, pp. 322--333, 2021.
	
	
	\bibitem{LuoSpatial}
	S. Luo, P. Yang, Y. Che, K. Yang, K. Wu, K. C. Teh, and S. Li, ``Spatial modulation for RIS-assisted uplink communication: Joint power allocation and passive beamforming design,'' \textit{IEEE Trans. Commun.}, vol.~69, no.~10, pp.~7017--7031, Oct.~2021.
	
	
	
	\bibitem{LiSingle}
	Q. Li, M. Wen, and M. Di Renzo, ``Single-RF MIMO: From spatial modulation to metasurface-based modulation,'' \textit{IEEE Wireless Commun.}, vol.~28, no.~4, pp.~88--95, Aug. 2021.	
	
	
	\bibitem{GuoReflecting}	
	S. Guo, S. Lv, H. Zhang, J. Ye, and P. Zhang, ``Reflecting modulation,'' \textit{IEEE J. Sel. Areas Commun.}, vol.~38, no.~11, pp.~2548--2561, Nov.~2020.
	
	
	\bibitem{YanPassive}	
	W. Yan, X. Yuan, Z. He, and X. Kuai, ``Passive beamforming and information transfer design for reconfigurable intelligent surfaces aided multiuser MIMO systems,'' \textit{IEEE J. Sel. Areas Commun.}, vol.~38, no.~8, pp.~1793--1808, Aug.~2020.
	
	
	\bibitem{BasarReconfigurable}
	E. Basar, ``Reconfigurable intelligent surface-based index modulation: A new beyond MIMO paradigm for 6G,'' \textit{IEEE Trans. Commun.}, vol. 68, no. 5, pp. 3187--3196, May~2020.
	
	\bibitem{DashPerformance}
		S. P. Dash, R. K. Mallik, and N. Pandey, ``Performance analysis of an index modulation-based receive diversity RIS-assisted wireless communication system,'' \textit{IEEE Commun. Lett.}, vol. 26, no. 4, pp. 768--772, Apr. 2022.
	
	
	\bibitem{YuanReceive}
	J. Yuan, M. Wen, Q. Li, E. Basar, G. C. Alexandropoulos, and G. Chen, ``Receive quadrature reflecting modulation for RIS-empowered wireless communications,'' \textit{IEEE Trans. Veh. Tech.}, vol.~70, no.~5, pp.~5121--5125, May~2021.
	

		\bibitem{GopiIntelligent}
		S. Gopi, S. Kalyani, and L. Hanzo, ``Intelligent reflecting surface assisted beam index-modulation for millimeter wave communication,'' \textit{IEEE Trans. Wireless Commun.}, vol.~20, no.~2, pp.~983--996, Feb.~2021.
		
		
		\bibitem{LinReconfigurable}
		S. Lin, B. Zheng, G. C. Alexandropoulos, M. Wen, M. Di Renzo, and F. Chen, ``Reconfigurable intelligent surfaces with reflection pattern modulation: Beamforming design and performance analysis,'' \textit{IEEE Trans. Wireless Commun.}, vol. 20, no. 2, pp. 741--754, Feb. 2021.
		
		\bibitem{LinReconfigurable2}
		S. Lin, F. Chen, M. Wen, Y. Feng, and M. Di Renzo, ``Reconfigurable intelligent surface-aided quadrature reflection modulation for simultaneous passive beamforming and information transfer,'' \textit{IEEE Trans. Wireless Commun.}, vol. 21, no. 3, pp. 1469--1481, Mar. 2022.
		

	
	
	\bibitem{YangIntelligent}
	Y. Yang, B. Zheng, S. Zhang, and R. Zhang, ``Intelligent reflecting surface meets OFDM: Protocol design and rate maximization,'' \textit{IEEE Trans. Commun.}, vol.~68, no.~7, pp.~4522--4535, Jul.~2020.
	
	\bibitem{ZhengIntelligent}
	B. Zheng and R. Zhang, ``Intelligent reflecting surface-enhanced OFDM: Channel estimation and reflection optimization,'' \textit{IEEE Wireless Commun. Lett.}, vol.~9, no.~4, pp.~518--522, Apr.~2020.
	
	\bibitem{LinAdaptive}
	S. Lin, B. Zheng, G. C. Alexandropoulos, M. Wen, F. Chen, and S. Mumtaz, ``Adaptive transmission for reconfigurable intelligent surface-assisted OFDM wireless communications,'' \textit{IEEE J. Sel. Areas Commun.}, vol.~38, no.~11, pp.~2653--2665, Nov.~2020.
	
	\bibitem{AlexandropoulosPhase}
	G. C. Alexandropoulos, S. Samarakoon, M. Bennis, and M. Debbah, ``Phase configuration learning in wireless networks with multiple reconfigurable intelligent surfaces,'' in \textit{Proc. IEEE Global Commun. Conf.}, Taipei, Taiwan, Dec.~2020, pp.~1--6.
	

		\bibitem{DashCoherent}
		S. P. Dash, R. K. Mallik, and S. K. Mohammed, ``Coherent detection in a receive diversity PLC system under Nakagami-$m$ noise environment,'' in \textit{Proc. IEEE Personal, Indoor and Mobile Radio Commun. Conf.}, Valencia, Spain, Sep.~2016, pp.~1--6.
		
		\bibitem{MallikA}
		R. K. Mallik, ``A new statistical model of the complex Nakagami-$m$ fading gain,'' \textit{IEEE Trans. Commun.}, vol.~58, no.~9, pp.~2611--2620, Sep.~2010.

	
	\bibitem{AlexandropoulosSwitch}
	G. C. Alexandropoulos, P. T. Mathiopoulos, and N. C. Sagias, ``Switch-and-examine diversity over arbitrary correlated Nakagami-$m$ fading channels,'' \textit{IEEE Trans. Veh. Tech.}, vol.~59, no.~4, pp.~2080--2087, May 2010.
	
	
	
	\bibitem{ChuPolyphase}
	D. C. Chu, ``Polyphase codes with good periodic correlation properties,'' \textit{IEEE Trans. Inf. Theory}, vol.~IT-18, no.~4, pp.~531--532, Jul.~1972.
	
\bibitem{ChianiNew}
		M. Chiani, D. Dardari, and M. K. Simon, ``New exponential bounds and approximations for the computation of error probability in fading channels,'' \textit{IEEE Trans. Wireless Commun.}, vol.~2, no.~4, pp.~840--845, Jul.~2003.
	
	\bibitem{WenCyclic}
	M. Wen, S. Lin, K. J. Kim, and F. Ji, ``Cyclic delay diversity with index modulation for green Internet of Things,'' \textit{IEEE Trans. Green Commun. Netw.}, vol.~5, no.~2, pp.~600--610, Jun.~2021.
	
	
	\bibitem{WenMultiple}
	M. Wen, E. Basar, Q. Li, B. Zheng, and M. Zhang, ``Multiple-mode orthogonal frequency division multiplexing with index modulation,'' \textit{IEEE Trans. Commun.}, vol.~65, no.~9, pp.~3892--3906, Sep.~2017.
	
	\bibitem{ZhangActive}
	Z. Zhang, L. Dai, X. Chen, C. Liu, F. Yang, R. Schober, and H. V. Poor, ``Active RIS vs. passive RIS: Which will prevail in 6G?'' \textit{arXiv preprint arXiv:2103.15154}, Mar. 2021.
	
	
\end{thebibliography}
\end{document}